\documentclass[a4paper,12pt]{article}

\usepackage{url}  
\usepackage[dvips]{color}
\usepackage{epsfig}
\usepackage{epstopdf}
\usepackage{subfigure}
\usepackage{float,latexsym,float,amsmath,amsthm,amssymb,amsfonts,lineno}
\usepackage[numbers,sort&compress]{natbib}
\usepackage{caption}
\usepackage[labelfont=bf]{caption}
\numberwithin{equation}{section}
\newtheorem{theorem}{Theorem}[section]

\newtheorem{corollary}{Corollary}[section]

\usepackage{textcomp}
\DeclareMathOperator*{\minimise}{minimise}
\begin{document}

\title{Mathematical analysis of a \textit{Wolbachia} invasive model with imperfect maternal transmission and loss of \textit{Wolbachia} infection}

\author{ Adeshina I. Adekunle $\footnote{corresponding Author, Email: adeshina.adekunle@jcu.edu.au}$, Michael T. Meehan, and Emma S. McBryde\\
{\it {\footnotesize Australian Institute of Tropical Health and Medicine, James Cook University, Australia}}
}

\maketitle

%\begin{linenumbers}

\begin{abstract}
Arboviral infections, especially dengue, continue to cause significant health burden in their endemic regions. One of the strategies to tackle these infections is to replace the main vector agent, \textit{Ae. aegypti}, with the ones incapable of transmiting the virus. \textit{Wolbachia}, an intracellular bacterium, has shown promise in achieving this goal. However, key factors such as imperfect maternal transmission, loss of \textit{Wolbachia} infection, reduced reproductive capacity and shortened life-span affect the dynamics of \textit{Wolbachia} in different forms in the \textit{Ae. aegypti} population.

In this study, we developed a \textit{Wolbachia} transmission dynamic model adjusting for imperfect maternal transmission and loss of \textit{Wolbachia} infection. The invasive reproductive number that determines the likelihood of replacement of the \textit{Wolbachia}-uninfected (WU) population is derived and with it, we established the local and global stability of the equilibrium points. This analysis clearly shows that cytoplasmic incompatibility (CI) does not guarantee establishment of the \textit{Wolbachia}-infected (WI) mosquitoes as imperfect maternal transmission and loss of \textit{Wolbachia} infection could outweigh the gains from CI. Optimal release programs depending on the level of imperfect maternal transmission and loss of \textit{Wolbachia} infection are shown. Hence, it is left to decision makers to either aim for replacement or co-existence of both populations.
\end{abstract}

\numberwithin{equation}{section}


\section{Introduction}
\label{S:1}

Vector borne diseases such as dengue, zika, chikungunya and yellow fever are of global health concern. For instance, dengue has a widespread geographical distribution with around 3.9 billion people at risk and an annual estimate of 390 million new dengue infections \cite{bhatt2013global, kyle2008global}. The major vector responsible for the transmission of dengue and other arboviral infections is the female \textit{Ae. aegypti} mosquito. Although, the female \textit{Ae. albopictus} mosquitoes also contributes. The risk of mortality due to dengue infection is low but is modified by the serotype of the infecting dengue and an individual's infection history (particularly their immune response to different serotypes) \cite{kyle2008global}. Despite the low risk of mortality, the large number of confirmed dengue cases and associated morbidity make dengue a substantial contributor to the global health burden.  The World Health Organization (WHO) global target for dengue by 2020 is to reduce morbidity and mortality by at least 25\% and 50\% respectively \cite{world2012global}. Integrated vector management is one among many potential control strategies being considered. Controlling the mosquito vectors appears to be promising but it comes with challenges and great cost \cite{manrique2015use, ooi2006dengue}.

Rather than preventing human-vector contacts, replacing the population of \textit{Ae. aegypti} mosquitoes with another variant that is incapable of viral transmission has been successfully applied to reduce dengue infections \cite{hoffmann2011successful} and the approach appears promising for other mosquito-borne infections such as chikungunya, malaria, West-Nile virus, and zika virus \cite{glaser2010native, dutra2016wolbachia, gomes2018infection, moreira2009wolbachia}. \textit{Wolbachia}, an intracellular insect bacterium, has the capacity to inhibit dengue virus proliferation inside the \textit{Ae. aegypti} mosquitoes and can spread via maternal (vertical) transmission \cite{moreira2009wolbachia, turley2009wolbachia}. Depending on the strain of \textit{Wolbachia},  mosquitoes infected with  \textit{Wolbachia} have a reproductive advantage over those uninfected via cytoplasmic incompatibility (CI) $-$ the mechanism that prevents the embryo maturing following mating between \textit{Wolbachia} infected (WI) males and \textit{Wolbachia} unfected (WU) females \cite{turelli1995cytoplasmic, ant2018wolbachia}. This advantage alone may not guarantee that WI mosquitoes will replace the \textit{Ae. aegypti} population, as \textit{Wolbachia} infection leads to a fitness cost to its host \cite{walker2011wmel, fine1978dynamics} and also, there are  reports regarding \textit{Wolbachia} infection and the loss of cytoplasmic incompatibility \cite{turelli1995cytoplasmic, fine1978dynamics, ross2019loss} as a result unfavorable conditions that lead to loss of \textit{Wolbachia} infection  in infected adults. Another factor that could prevent \textit{Wolbachia} infected mosquitoes from dominating the \textit{Ae. aegypti} population is imperfect maternal transmission \cite{turelli1995cytoplasmic, fine1978dynamics, turelli2010cytoplasmic, yeap2011dynamics}. Hence, having a full understanding of the interplay between key parameters in \textit{Wolbachia} introduction is necessary to ensure the success of the strategy if it is to be used on a large scale.

Mathematical modelling plays a significant role in understanding the impact of variables involved in the dynamics of a particular infectious disease and has been used in the decision-making process that guides the application of some typical control strategies \cite{martcheva2015introduction}. Different mathematical models have been developed to simulate the introduction of \textit{Wolbachia} into \textit{Ae. aegypti} populations \cite{ndii2012modelling, xue2017two, crain2011wolbachia, schraiber2012constraints, zheng2018wolbachia, campo2018optimal, Rafikov2019}, with each specifying conditions that enable WI mosquitoes to dominate. Caspari and Watson \cite{caspari1959evolutionary} demonstrated the importance of cytoplasmic incompatibility on the population replacement between WU and WI mosquitoes. Ndii et al. developed a deterministic compartmental model for the competition between the two mosquitoes populations and derived the steady-state solutions showing key parameters that could influence the competition between the two populations \cite{ndii2012modelling}. Xue et al. adopted similar appproach as \cite{ndii2012modelling} by incorporating sex structure into the compartmental models and showed that the endemic \textit{Wolbachia} steady-state solution can be established by releasing a sufficiently large number of \textit{Wolbachia} infected mosquitoes \cite{xue2017two}.  Using ordinary differential equations to model the competitions between WU and WI mosquitoes, Zhang et al. showed that the successful replacement of WU mosquitoes with WI ones would depend on the strains of \textit{Wolbachia} used and require a careful release design \cite{zheng2018wolbachia}. The idea of designed release methods was further emphasised by Qu et al. when they extend the model due to \cite{xue2017two} to include the fact that most female mosquitoes mate once \cite{Qu2018}. The model of Li and Liu places emphasis on the combinations of birth and death rate functions, \textit{Wolbachia} strain and the number of WI mosquitoes released \cite{Li2017}. All these modelling works pointed to the possibility of WI mosquitoes replacing the uninfected ones.

In this paper, we consider the impacts of imperfect maternal transmission and loss of \textit{Wolbachia} infection by investigating the asymptotic dynamics of the \textit{Wolbachia} invasive model and determining the necessary and sufficient conditions for \textit{Wolbachia} invasion. These two factors (imperfect maternal transmission and loss of  \textit{Wolbachia} infection) have not been considered by previous models and the derivation of both asymptotic and global stability of the possible equilibrium points are the novel results of this work. With the derivation of the global stability, appropiate control stategies can be adopted to ensure that WI mosquitoes can replace uninfected ones or at least become more abundant than the uninfected ones. We consider these strategies via optimal control.
\section{Model formulation}
 We consider the \textit{Ae. aegypti} mosquito population that is responsible for the transmission of most arboviral infections, in particular dengue virus. Similar to \cite{xue2017two}, the \textit{Ae. aegypti} mosquito population is divided into two major subpopulations: those with \textit{Wolbachia} infection ($w$); and those without \textit{Wolbachia} infection ($\bar{w}$). We denote the number of mosquitoes that are wildtype (i.e WU) in the aquatic stage (egg, larvae, and Pupae), adult male and adult female stage as $Q_{\bar{w}},M_{\bar{w}}$ and$\ F_{\bar{w}}$, respectively, and those with \textit{Wolbachia} infection as $Q_w,M_w$ and${\ F}_w$. If we assume a logistic growth in the aquatic stage (egg, larvae and Pupae) \cite{ndii2012modelling, xue2017two}, the dynamics of the \textit{Ae. aegypti} mosquito population are modelled as:

\begin{equation} \label{Eq1} 
\underbrace{\frac{\ dQ_{\bar{w}}}{dt\ }=\left[\frac{{\phi }_{\bar{w}}F_{\bar{w}}M_{\bar{w}}+\rho_{_1}{\phi }_wF_wM_w+\rho_{_2}{\phi}_wF_wM_{\bar{w}}}{M_{\bar{w}}{+M}_w}\right]
\left(1-\frac{Q}{K}\right)_{+}-({\mu }_a+\psi {)Q}_{\bar{w}}}_{\mathrm{Wildtype\ aquatic\ stage}},        
\end{equation} 
\begin{equation} \label{Eq2} 
\underbrace{\frac{dQ_w}{dt}=\left[\frac{{(1-\rho_{_1})\phi }_wF_wM_w+(1-\rho_{_2}){\phi }_wF_wM_{\bar{w}}}{M_{\bar{w}}{+M}_w}\right]\left(1-\frac{Q}{K}\right)_{+}-{(\mu }_a+\psi \mathrm{\ )}Q_w}_{\mathit{Wolbachia}\mathrm{\ aquatic\ stage}}, 
\end{equation} 
\begin{equation} \label{Eq3} 
\underbrace{\frac{dF_{\bar{w}}}{dt}=b\psi Q_{\bar{w}}+\sigma F_w-{\mu }_{\bar{w}}F_{\bar{w}}}_{WU\ female\ adult}, 
\end{equation} 
\begin{equation} \label{Eq4} 
\underbrace{\frac{dM_{\bar{w}}}{dt}=\left(1-b\right)\psi Q_{\bar{w}}+\sigma M_w-{\mu }_{\bar{w}}M_{\bar{w}}}_{WU\ male\ adult},             
\end{equation} 
\begin{equation} \label{Eq5} 
\underbrace{\frac{dF_w}{dt}=b\psi Q_w-\sigma F_w-{\mu }_w{\ F}_{w\ },}_{WI\ female\ adult} 
\end{equation} 
\begin{equation} \label{Eq6} 
\underbrace{\frac{dM_w}{dt}=\left(1-b\right)\psi Q_w-\sigma M_w-{\mu }_wM_w,}_{WI\ male\ adult}
\end{equation}
where  $Q=Q_{\bar{w}}+Q_w$ is the total number of aquatic stage mosquitoes which we assumed is less than the carrying capacity $(K)$. That is $\left(1-\frac{Q}{K}\right)_{+} = \max\left(0,\left(1-\frac{Q}{K}\right)\right)$. Unlike the \textit{Wolbachia} invasion model in \cite{xue2017two}, we assumed that the proportions of offspring due to imperfect maternal transmission are different between adult \textit{Ae . aegypti} males and females of different \textit{Wolbachia} infection status \cite{turelli1995cytoplasmic, Bian2013}.  We further include the possibility of a decline in the level of \textit{Wolbachia} infection  by allowing some WI  to become WU at a constant per capital rate $\sigma$.

The system of differential equations (\ref{Eq1} - \ref{Eq6}) is very complex. Hence, we reduce it to a simpler model that preserves the key dynamic features. This will enable us to study the dynamics of the \textit{Wolbachia} replacement strategy. One way to do this is to assume an equal number of male and female mosquitoes. This is reasonable considering the experimental work in \cite{arrivillaga2004} which estimated the ratio of male to female mosquitoes as $1.02 : 1$. Hence, by setting $M=F$ (for both $w$- and $\bar{w}$-type mosquitoes) and $b=\frac{1}{2}$, the system of differential equations  (\ref{Eq1} - \ref{Eq6}) can be reduced to:
\begin{align} \label{Eq7} 
\frac{\ dQ_{\bar{w}}}{dt\ }=&\left[\frac{{\phi }_{\bar{w}}F^2_{\bar{w}}+{{\rho_{_1}\phi }_wF}^2_w+\rho_{_2}{\phi }_wF_wF_{\bar{w}}}{F_{\bar{w}}+F_w}\right]\left(1-\frac{Q}{K}\right)-{(\mu }_a+\psi )Q_{\bar{w}}, \\\label{Eq8} 
\frac{dQ_w}{dt}=&\left[\frac{{{(1-\rho_{_1})\phi }_wF}^2_w+(1-\rho_{_2}){\phi }_wF_wF_{\bar{w}}}{F_{\bar{w}}+F_w}\right]\left(1-\frac{Q}{K}\right)-{(\mu }_a+\psi )Q_w,\\ \label{Eq9} 
\frac{dF_{\bar{w}}}{dt}=&\frac{\psi }{2}Q_{\bar{w}}+\sigma F_w-{\mu }_{\bar{w}}F_{\bar{w}}, 
\\\label{Eq10} 
\frac{dF_w}{dt}=&\frac{\psi }{2}Q_w-\sigma F_w-{\mu }_wF_w, 
\end{align} 
with the understanding that $\left(1-\frac{Q}{K}\right)$ is always non-negative. This system of differential equations (\ref{Eq7} - \ref{Eq10}) explicitly includes the possibility of uninfected and infected offspring being produced by WI female mosquitoes \cite{hoffmann2011successful, turelli1995cytoplasmic, turelli2010cytoplasmic, yeap2011dynamics}, through the ${\rho_{_1}}$and ${\rho_{_2}}$ terms. This possibility is excluded in the model of \textit{Wolbachia} introduction studied in \cite{ndii2012modelling}. However, they adjusted for leakage by assuming there is a waiting time before WI offspring mature into either WI or WU adults. It is biologically plausible that some offspring are born uninfected by WI females \cite{turelli1995cytoplasmic} and the approach by \cite{ndii2012modelling} is another way of modelling the effect of losing incompatibility between WU females and WI males, which we have incorporated with the $\sigma $ term (see Figure \ref{Fig1}).  The descriptions of the parameters in the system of differential equations (\ref{Eq7} - \ref{Eq10}) are shown in Table \ref{Table1}.

We analyse the \textit{Wolbachia} invasive system (\ref{Eq7} - \ref{Eq10}) for the conditions that will enable the WI mosquitoes to propagate following their introduction into an \textit{Ae. aegypti} population that is na\"{i}ve to \textit{Wolbachia} infection. This is done as follows. Given a system of autonomous ordinary differential equations,
\begin{equation}\label{Eq11}
\frac{d\boldsymbol{X}}{dt}=f\left(\boldsymbol{X}\right),\text{where} \ \boldsymbol{X}\boldsymbol{,\ }f\left(\boldsymbol{X}\right)\in {\mathbb{R}}^n,  
\end{equation} 
the asymptotic behavior of the solutions of (\ref{Eq11}) starting near an equilibrium solution $\bar{\boldsymbol{X}}$ are determined by the eigenvalues of the associated Jacobian matrix defined as $\boldsymbol{J}\boldsymbol{=}\frac{\partial f}{\partial \boldsymbol{X}}$ evaluated at $\bar{\boldsymbol{X}}$\textbf{.}  We adopt this approach in this paper to understand \textit{Wolbachia} propagation in the \textit{Ae. aegypti} population. Also, the simulations of this model using published parameter values were done in MATLAB R2017a (Release M(2017) The MathWorks Inc, Natick, MA, USA).

\begin{figure}[h]
\begin{center}
\includegraphics[width=0.6\linewidth]{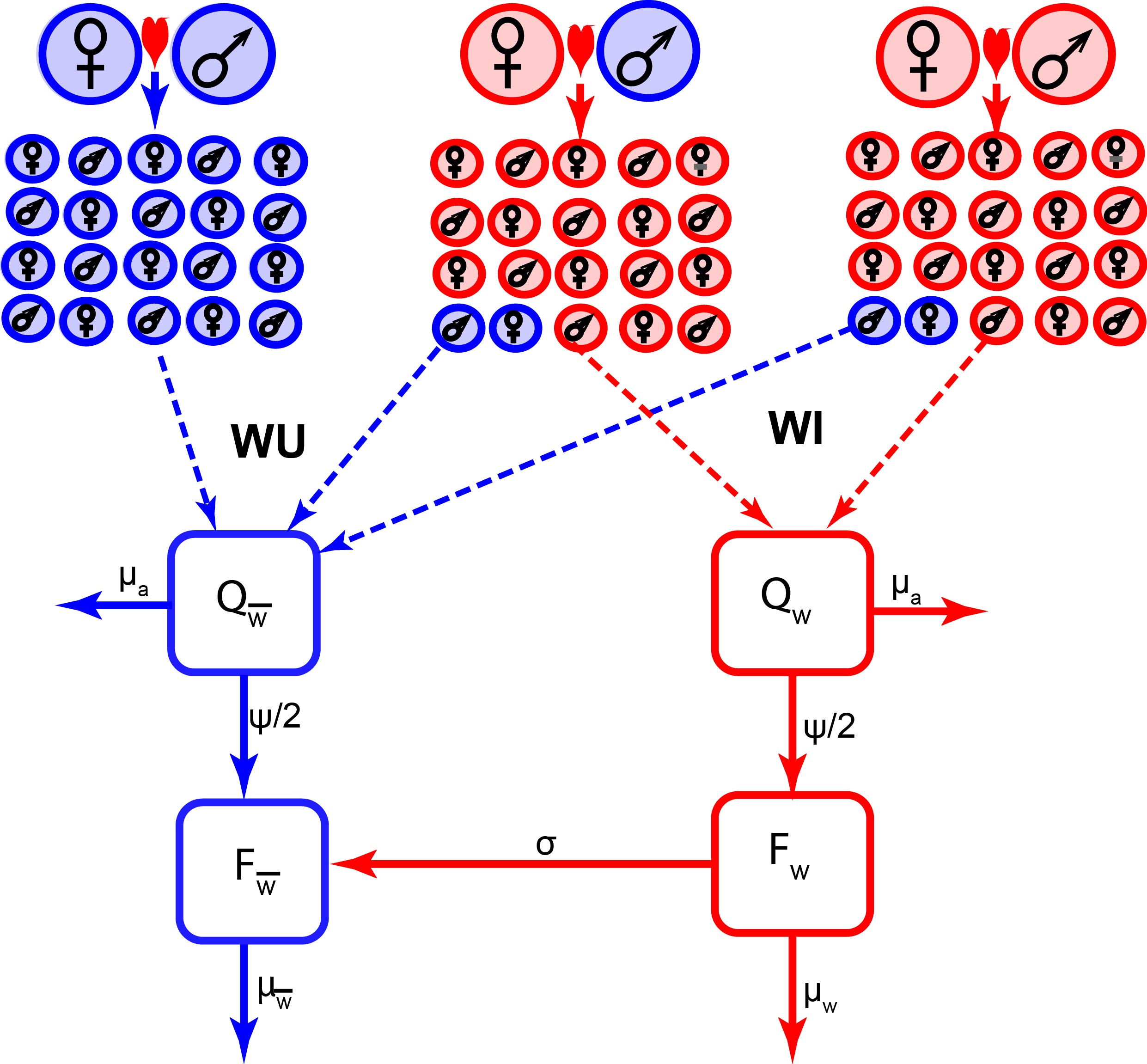}
\caption{\textbf{Schematic flow diagram of the \textit{Wolbachia} invasive model.} The WU population (blue color) is produced as a result of mating between adult WU females and males, WI females and WU males, WI females and males and loss of \textit{Wolbachia} infection by adult WI females. The WI population (red color) is produced by mating between adult WI females and males and cross-breeding between WI females and WU males. Due to CI, the offsprings as a result of mating between adult WU females and WI males are not viable.}
\label{Fig1}
\end{center}
\end{figure}
\begin{table}
\begin{center}
\caption{Parameter description and values for model (\ref{Eq7} - \ref{Eq10})}
\begin{tabular}{p{0.6in}p{1.6in}p{1.2in}p{1.0in}p{0.8in}} \hline 
Parameter & Description & Estimate[Range] & Unit & References \\ \hline 
$K$ & Carrying capacity of the aquatic stage & ${10}^6 [{10}^4,{10}^8]$ & Aquatic mosquito & Assumed \\
${\phi }_{\bar{w}}$ & Per capita egg laying rate for \textit{Wolbachia}
uninfected mosquitoes & $13 [12 - 18]$ & Eggs per day & \cite{Hoffmann2014, McMeniman2009,McMeniman2010} \\  
${\phi }_w$ & Per capita egg laying rate for \textit{Wolbachia-}infected mosquitoes & $11 [8 - 12]$ & Eggs per day & \cite{walker2011wmel, Hoffmann2014} \\ 
${\rho_{_1}}$ & The fraction of eggs that are WU as a result of mating between adult WI female and male mosquitoes & $0.05 [0-0.11]$ & Dimensionless & \cite{walker2011wmel} \\  
${\rho_{_2}}$ & The fraction of eggs that are WU as a result of mating between adult WU male and WI female mosquitoes & $0.05 [0- 0.1]$ & Dimensionless & \cite{walker2011wmel} \\ 
$\sigma $ & Per capita loss of \textit{Wolbachia} infection & $0.04 [0 - 0.1]$ & Per day & Assumed \\  
$b$ & Fraction of eggs that are female & $0.5 [0.34-0.6]$ & Dimensionless & \cite{arrivillaga2004,Lounibos2008} \\  
$\psi $ & Per capita maturation rate & $0.11 [0.1 -0.12]$  & Per day & \cite{walker2011wmel,Hoffmann2014} \\ 
${\mu }_a$ & Per capita aquatic death rate & $0.02$ & Per day & \cite{xue2017two} \\ 
${\mu }_{\bar{w}}$ & Per capita death rate of WU mosquitoes & $0.061 [0.02-0.09]$  & Per day & \cite{McMeniman2009, Styer2007} \\ 
${\mu }_w$ & Per capita death rate of WI mosquitoes & $0.068 [0.03-0.14]$ & Per day &  \cite{walker2011wmel, Styer2007} \\ \hline 
\end{tabular}
\label{Table1}
\end{center}
\end{table}
\newpage
\section{\textbf{\textit{Wolbachia} invasive model with $\mathbf{\rho_{_1}=0}$ and $\mathbf{\sigma =0}$}}
When $\rho_{1} = 0$ and $\sigma = 0$ in equations \eqref{Eq7} - \eqref{Eq10}, we assumed imperfect maternal transmission is only between WI females and WU males and that there is no loss of \textit{Wolbachia} infection in adult infected mosquitoes.  The \textit{Wolbachia} invasive model (\ref{Eq7} - \ref{Eq10}) with $\rho_{1} = 0$ and $\sigma = 0$ is biologically meaningful (see \ref{ap1}). That is, all solutions with non-negative initial conditions will remain non-negative for future times. Showing this for $\rho_1 \in $ (0, 1] and $\sigma > 0$ is also straight-forward.

The \textit{Wolbachia} invasive model \eqref{Eq7} - \eqref{Eq10} with $\rho_{_1}=0$ and $\sigma =0$ has four steady states: $E_1=(0, 0, 0, 0)$ - where there are no mosquitoes; $E_2=(Q^*_{\bar{w}},0,\ F^*_{\bar{w}},\ 0)$ - where the WU mosquitoes dominate and leads to the extinction of infected ones; $E_3=(0,Q^*_w,0, F^*_w)$ - where only the WI mosquitoes exist; and $E_4=(Q^*_{\bar{w}},\ Q^*_w,\ F^*_{\bar{w}},\ F^*_w)$ - where both WU and WI mosquitoes coexist. It is important for the control of arboviral infections that are transmitted by \textit{Ae. aegypti} mosquitoes to determine the nature of these stability points.

\subsection{No mosquitoes}
The $E_1$ point is trivial but not interesting as it is not realistic.  However, we can gain insights about the nature of this steady state solution by examining a special case when there is no interaction between WU and WI mosquitoes. We derived         
\begin{equation} \label{Eq12} 
R_{0\bar{w}}=\frac{{\phi }_{\bar{w}}\psi }{{2\mu }_{\bar{w}}\left({\mu }_a+\psi \right)},
\end{equation} 
and 
\begin{equation} \label{Eq13} 
R_{0w}=\frac{{\phi }_w\psi }{2{\mu }_w\left({\mu }_a+\psi \right)},
\end{equation} 
which are the thresholds that determine whether each population will persist or extinguish in the absence of interactions. The thresholds in equations \eqref{Eq12} and \eqref{Eq13} are derived from the stability conditions of the associated Jacobian matrix when no interaction exists between the uninfected and infected mosquitoes. That is, individual populations do not depend on each other. Equivalent expressions were given in \cite{xue2017two} for the dynamics that explicitly include the male mosquito compartments. Hence, for these models \eqref{Eq7}-\eqref{Eq10}, the two populations are extinguished whenever $R_{0\bar{w}}<1\ $and $R_{0w}<1,$ (see Figure \ref{Fig2}) as the reproductive terms cannot sustain the populations. Also, since the solutions are always non-negative for non-negative initial data, the solutions tend to the no-mosquito equilibrium point. However and except for the biological implications of using insecticides, appling insecticides and destroying breeding sites have been effective method in reducing mosquito populations \cite{amer2006larvicidal}. 
\begin{figure}[h]
\begin{center}
\includegraphics[width =0.8\linewidth]{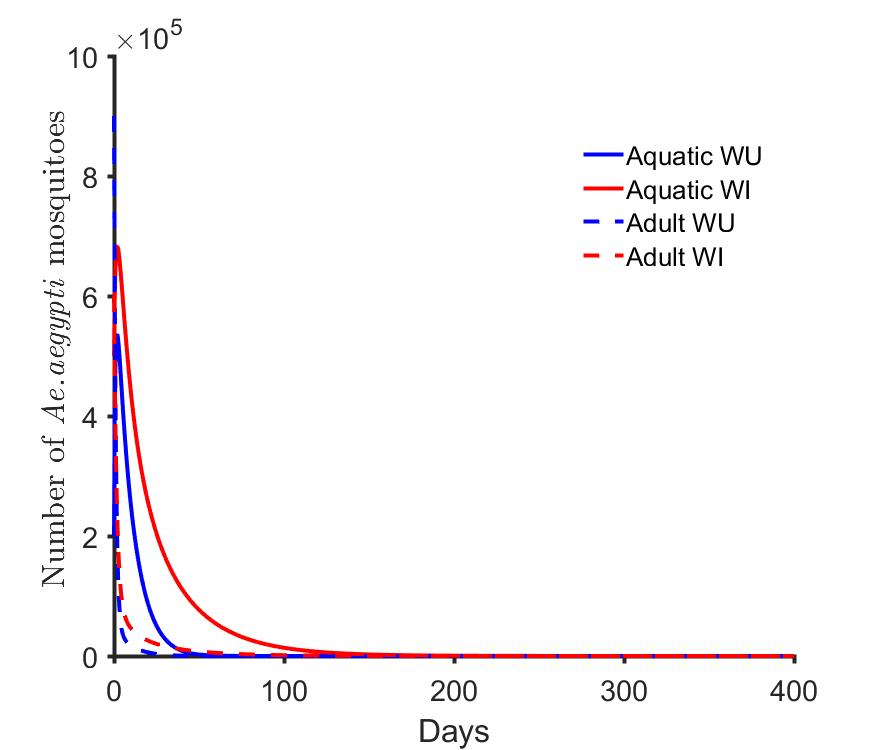}
\caption{\textbf{No mosquito equilibrium point}. In this simulation, we set $\phi_{w} =2$, $\phi_{\bar{w}}=1$, $R_{0w} = 0.71$, $R_{0\bar{w}} = 0.94$, $K = 2000000$, $Q_{\bar{w}}(0) = 200000$, $Q_{w}(0) = 500000$, $F_{\bar{w}}(0) = 900000$, and $F_w(0) = 600000$.}
\label{Fig2}
\end{center}
\end{figure}
\newpage
\subsection{WU mosquitoes-only}\label{sub1}
The WU equilibrium point is given as $E_2=\left(Q^*_{\bar{w}},0,\ F^*_{\bar{w}},\ 0\right)$ where
\begin{equation} \label{Eq14} 
Q^*_{\bar{w}}= K\left(1-\frac{1}{R_{0\bar{w}}}\right), 
\end{equation} 
\begin{equation} \label{Eq15} 
\ F^*_{\bar{w}}=\frac{\psi K}{2{\mu }_{\bar{w}}}\left(1-\frac{1}{R_{0\bar{w}}}\right), 
\end{equation} 
Hence  $R_{0\bar{w}}>1$ is necessary for the existence of this equilibrium point, otherwise, there will be no WU mosquitoes. Let us define the invasive reproductive number $\left(R_{0w|\bar{w}}\right)$ associated with the number of secondary offspring that would be WI due to the introduction of a typical WI adult mosquito into a population of WU adult mosquitoes. In a similar way to \cite{xue2017two}, we derived  $R_{0w|\bar{w}}$  as follows:
the WI compartments can be divided into the rate of appearance of new mosquitoes with \textit{Wolbachia} infection ($\mathcal{F}$) and other transition rates such as progression into adult mosquitoes with \textit{Wolbachia} infection and death rates $\left(\mathcal{V}\right):$
\begin{equation} \label{Eq16} 
\mathcal{F}=\left( \begin{array}{c}
\left(\frac{{\phi }_wF^2_w+\left(1-\rho_{_2}\right){\phi }_wF_wF_{\bar{w}}}{F_{\bar{w}}+F_w}\right)\left(1-\frac{Q}{K}\right) \\ 
0 \end{array}
\right), 
\end{equation} 
\begin{equation} \label{Eq17} 
\mathcal{V}=\left( \begin{array}{c}
\left({\mu }_a+\psi \right)Q_w \\ 
-\frac{\psi Q_w}{2}+{\mu }_wF_w \end{array}
\right). 
\end{equation} 

Next, we introduce the matrices $F$ and $V$ with components $F_{ij}={\left.\frac{\partial {\mathcal{F}}_i}{\partial x_j}\right|}_{E_2}$ and $V_{ij}={\left.\frac{\partial {\mathcal{V}}_i}{\partial x_j}\right|}_{E_2}$, where the $x_js$ represent the infected compartments $Q_w$ and $F_w$. Hence,
\begin{equation} \label{Eq18} 
F=\left( \begin{array}{cc}
0 & \frac{{\phi }_w(K-{Q_{\bar{w}}}^*)(1-\rho_{_2})}{K} \\ 
0 & 0 \end{array}
\right), 
\end{equation} 
\begin{equation} \label{Eq19} 
V=\left( \begin{array}{cc}
\mu_a + \psi  & 0 \\ -\frac{\psi }{2} & {\mu }_w \end{array}
\right) 
\end{equation} 
and the next-generation matrix is 
\begin{equation} \label{Eq20} 
FV^{-1}=\left( \begin{array}{cc}
\frac{{\psi \phi }_w(K-{Q_{\bar{w}}}^*)(1-\rho_{_2})}{2({\mu }_a+\psi ){\mu }_wK} & \frac{{\phi }_w(K-{Q_{\bar{w}}}^*)(1-\rho_{_2})}{{\mu }_wK} \\ 
0 & 0 \end{array}
\right). 
\end{equation} 
Hence the invasive reproductive number is 
\begin{equation} \label{Eq21} 
R_{0w|\bar{w}}=\lambda(FV^{-1})=\frac{{\phi }_w{\mu }_{\bar{w}}(1-\rho_{_2})}{{\phi }_{\bar{w}}{\mu }_w}=\frac{R_{0w}(1-\rho_{_2})}{R_{0\bar{w}}}, 
\end{equation} 
where $\lambda(M)$ is the spectral radius of $M$. The factor $(1-\rho_{_2} )$ shows the effect of the proportion of aquatic stage mosquitoes that are WI as a result of mating between WU male mosquitoes and WI female mosquitoes on the likelihood of the WI mosquitoes to replace the WU ones.

The Jacobian at $E_2$ is given as
\begin{equation} \label{Eq22} 
J\left(E_2\right)=\left( \begin{array}{cccc}
-\left({\mu }_a+\psi \right)R_{0\bar{w}} & \left({\mu }_a+\psi \right)(1-R_{0\bar{w}}) & \frac{\phi_{\bar{w}}}{R_{0\bar{w}}} &\frac{\left({\rho_{_2} {\phi }_w-\phi }_{\bar{w}}\right)}{R_{0\bar{w}}}\\
0 & -\left({\mu }_a+\psi \right)& 0 & \frac{{\phi }_w\left(1-\rho_{_2} \right)}{R_{0\bar{w}}}\\
\frac{\psi }{2}&  0  & -{\mu }_{\bar{w}}& 0 \\
0 & \frac{\psi }{2} & 0 & -{\mu }_w \end{array}
\right). 
\end{equation} 
The characteristic equation for this Jacobian is given as 
\begin{equation} \label{Eq23} 
P\left(\lambda \right):=({\lambda }^2+a_1\lambda +a_2)({\lambda }^2+b_1\lambda +b_2)=0, 
\end{equation} 
where
\begin{equation} \label{Eq24} 
a_1=\frac{2{\mu }^2_{\bar{w}}+{\phi }_{\bar{w}}\psi }{2{\mu }_{\bar{w}}}>0, 
\end{equation} 
\begin{equation}\label{Eq25}
a_2=\mu _{\bar{w}}({\mu }_a+\psi )(R_{0\bar{w}}-1),
\end{equation} 
\begin{equation} \label{Eq26} 
b_1={\mu }_a+\psi +{\mu }_w>0,              
\end{equation} 
\begin{equation}\label{Eq27}
b_2=\mu _w({\mu }_a+\psi )(1-R_{0w|\bar{w}}).                                                                               
\end{equation} 
The equilibrium point $E_2$ is locally asymptotically stable whenever $R_{0w|\bar{w}}<1$ and $R_{0\bar{w}}>1$. This implies for this case that WI mosquitoes will not spread following their introduction if those conditions are satisfied. The condition $R_{0\bar{w}}>1$ is the same as the only condition given by Ndii et al. \cite{ndii2012modelling} for the stability of this point. We have additional condition ($R_{0w|\bar{w}}<1$) which states that the invasive reproductive rate of WI mosquitoes when introduced into a background of WU mosquitoes be less than one.

\subsection{WI mosquitoes-only}
The equilibrium point associated with WI\textit{ }mosquitoes only is\textbf{}
\begin{equation} \label{Eq28} 
E_3=\left(0,K\left(1-\frac{1}{R_{0w}}\right),0,\frac{\psi K}{2{\mu }_w}
 \left( 1-\frac{1}{R_{0w}}\right)\right).                                                          
\end{equation} 
As pointed out earlier that WI population dies out when  $R_{0w}\le 1$ , the equilibrium point ($E_3$) is expected to be unstable when${\ R}_{0w}\le 1$. Thus, the corresponding Jacobian is defined as
\begin{equation} \label{Eq29} 
J\left(E_3\right)=\left(\begin{array}{cccc}
 -\left({\mu }_a+\psi \right) & 0 &\frac{\rho_{_2} {\phi }_w}{R_{0w}} & 0\\
 \left({\mu }_a+\psi \right)({1-R}_{0w}) & -\left({\mu }_a+\psi \right)R_{0w} & -\frac{(1-\rho_{_2} )w}{R_{0w}}  & \frac{{\phi }_w}{R_{0w}}\\
\frac{\psi }{2} & 0 & -{\mu }_{\bar{w}} & 0\\ 
0 & \frac{\psi }{2}  &  0  & -{\mu }_w
 \end{array}
\right). 
\end{equation} 
Hence, the characteristic equation is
\[P\left(\lambda \right):=({\lambda }^2+c_1\lambda +c_2)({\lambda }^2+d_1\lambda +d_2)=0,\] 
where $c_1=\frac{{2\mu }^2_w+{\phi }_w\psi }{{2\mu }_w}>0$, $c_2={\mu }_w\left({\mu }_a+\psi \right)\left(R_{0w}-1\right)$, $d_1={\mu }_a+\psi +{\mu }_{\bar{w}}>0$, and $d_2=({\mu }_a+\psi )({\mu }_{\bar{w}}-\rho_{_2} {\mu }_w)$.
The eigenvalues of the quartic characteristic equation are negative or have negative real parts if $R_{0w}>1$ and ${\mu }_{\bar{w}}>\rho_{_2} {\mu }_w$. The condition ${\mu }_{\bar{w}}>\rho_{_2} {\mu }_w$ gives the fitness level of the adult WI that is sufficient for spread. As \textit{Wolbachia} infection decreases the fitness of infected mosquitoes, external support will be needed for WI mosquitoes to propagate \cite{walker2011wmel}. 

Following from the expression of the basic reproduction number (equation (\ref{Eq21})), the equilibrium point for adult female mosquitoes can be written as 
\begin{equation} \label{Eq30} 
F^{*}_{w }= \frac{\psi K}{{2\mu }_w}\left(1-\frac{(1-\rho_{_2} )}{{R_{0w|\bar{w}}R}_{0\bar{w}}}\right). 
\end{equation} 

The expression above shows that the WI-mosquito-only equilibrium can exist when ${R}_{0\bar{w}|w}<1$. 
The existence of endemic equilibria for${\ R}_{0w|\bar{w}}<1$ is an indicator of a backward bifurcation in the conventional infectious diseases modelling papers \cite{cui2008impact, dushoff1998backwards}. However, it is interesting to know that this equilibrium point is unstable whenever $R_{0w|\bar{w}}<\frac{\left(1-\rho_{_2} \right)\ }{R_{0\bar{w}}}(\Rightarrow R_{0w}<1)$ and it is locally asymptotically stable even if $R_{0w|\bar{w}}<1$, in as much as $R_{0w}>1$ and ${\mu }_{\bar{w}}>\rho_{_2} {\mu }_w$. For $R_{0w|\bar{w}}<1$, both the $E_2$ and $E_3$ equilibrium points are locally asymptotically stable in as much as $R_{0\bar{w}}>1$ for $E_2$, and $R_{0w}>1$ and ${\mu }_{\bar{w}}>\rho_{_2}{\mu }_w$ for $E_3$.
\begin{theorem}\label{th1}
Provided ${\mu }_{\bar{w}}>\rho_{_2}{\mu }_w$, the WI\textit{ } mosquitoes only equilibrium point $E_3$ is globally asymptotically stable whenever $R_{0w|\bar{w}}>1$ and $ R_{0w}\ge R_{0w}^*>1\ $, 
\end{theorem}
\begin{proof}
When $R_{0w|\bar{w}}>1$ then $R_{0w}>\frac{R_{0\bar{w}}}{1-\rho_{_2}}$ and this implies$R_{0w}>R_{0\bar{w}}$. Define a Lyapunov function $V$ as
\begin{equation} \label{Eq31} 
V=\frac{\psi}{2\mu_{w}\left(\mu_a +\psi \right)} \int^{Q_{w}}_{Q^*_{w}}{\left(1-\frac{Q^*_w}{y}\right)}dy+\frac{1}{\mu _w}\int^{F_{w}}_{F^*_{w}}{\left(1-\frac{F^*_{w}}{y}\right)}dy. 
\end{equation} 
Differentiating equation \eqref{Eq31} with respect to time, we have
\begin{equation} \label{Eq32} 
\frac{dV}{dt}=\frac{\psi }{2{\mu }_w\left({\mu }_a+\psi \right)}\left(1-\frac{Q^*_w}{Q_w}\right)\frac{dQ_w}{dt}+\ \frac{1}{{\mu }_w}\left(1-\frac{F^*_w}{F_w}\right)\frac{dF_w}{dt}. 
\end{equation} 
Substituting the expression for the differential equations \eqref{Eq8} and \eqref{Eq10} we have,

\begin{align}\nonumber 
\frac{\psi }{{2\mu }_w\left({\mu }_a+\psi \right)}\left(1-\frac{Q^*_w}{Q_w}\right)\frac{dQ_w}{dt}=&\frac{\psi }{{2\mu }_w\left({\mu }_{AV}+\psi \right)}\left(1-\frac{Q^*_w}{Q_w}\right)\\\label{Eq33}
&\left(U_1\left(1-\frac{Q}{K}\right)-{(\mu }_a+\ \psi )Q_w\right) ,  
\end{align} 
where  $U_1=\left(\frac{{\phi }_wF^2_w+\left(1-\rho_{_2}\right){\phi }_wF_wF_{\bar{w}}}{F_{\bar{w}}+F_w}\right)$, and
\begin{equation} \label{Eq34} 
\frac{1}{{\mu }_w}\left(1-\frac{F^*_w}{F_w}\right)\frac{dF_w}{dt}=\frac{1}{{\mu }_w}\left(1-\frac{F^*_w}{F_w}\right)\left(\frac{\psi Q_w}{2}-{\mu }_wF_w\right). 
\end{equation} 
From equation \eqref{Eq33},
\begin{align}\nonumber
\frac{\psi }{{2\mu }_w\left({\mu }_a+\psi \right)}\left(1-\frac{{Q_w}^*}{Q_w}\right)\frac{dQ_w}{dt}=&\frac{\psi }{{2\mu }_w\left({\mu }_a+\psi \right)}\left(1-\frac{Q^*_w*}{Q_w}\right)U_1\left(1-\frac{Q}{K}\right)\\ \label{Eq35} 
&-\frac{\psi Q_w}{{2\mu }_w}+\frac{\psi Q^*_w}{{2\mu }_w} 
\end{align} 
and from equation \eqref{Eq34},
\begin{equation} \label{Eq36} 
\frac{1}{{\mu }_w}\left(1-\frac{F^*_w}{F_w}\right)\frac{dF_w}{dt}=\frac{\psi Q_w}{{2\mu }_w}-\frac{\psi Q_wF^*_w}{2{\mu }_wF_w}-F_w+F^*_w. 
\end{equation} 
Adding equations \eqref{Eq35} and \eqref{Eq36} yields
\begin{equation} \label{Eq37} 
\frac{dV}{dt}=\frac{\psi }{{2\mu }_w\left({\mu }_a+\psi \right)}\left(1-\frac{Q^*_w}{Q_w}\right)U_1\left(1-\frac{Q}{K}\right)+\frac{\psi Q^*_w}{{2\mu }_w}-\frac{\psi Q_wF^*_w}{2{\mu }_wF_w}-F_w+F^*_w. 
\end{equation} 
Rearrangement and some manipulations give,
\begin{align}\nonumber
\frac{dV}{dt}&=\frac{\psi }{{2\mu }_w\left({\mu }_a+\psi \right)}\left(1-\frac{Q^*_w}{Q_w}\right)U_1\left(1-\frac{Q}{K}\right)+F^*_w\left(2- \frac{Q_wF^*_w}{Q^*_wF_w}-\frac{Q^*_wF_w}{Q_wF^*_w}\right)\\ \label{Eq38} 
&-F_w\left(1-\frac{Q^*_w}{Q_w}\right).  
\end{align} 
Thus,
\begin{align}\nonumber
\frac{dV}{dt}=&R_{0w}F_w\left(1-\frac{Q^*_w}{Q_w}\right)\left(\frac{F_w+\left(1-\rho_{_2}\right)F_{\bar{w}}}{F_w+F_{\bar{w}}}\left(1-\frac{Q}{K}\right)-\frac{1}{R_{0w}}\right)\\ \label{Eq39} &+F^*_w\left(2- \frac{Q_wF^*_w}{Q^*_wF_w}-\frac{Q^*_wF_w}{Q_wF^*_w}\right), 
\end{align} 
\begin{align}\nonumber
\frac{dV}{dt}= &F_w\left(1-\frac{Q^*_w}{Q_w}\right)\left(\frac{R_{0w}(F_w+\left(1-\rho_{_2}\right)F_{\bar{w}})}{F_w+F_{\bar{w}}}\left(1-\frac{Q}{K}\right)-1\right)\\ \label{Eq40} 
&+F^*_w\left(2- \frac{Q_wF^*_w}{Q^*_wF_w}-\frac{Q^*_wF_w}{Q_wF^*_w}\right). 
\end{align} 
Since $Q_w \le K$, $\left(1-\frac{Q^*_w}{Q_w}\right)<0$ when $Q_w < Q^*_{w}$ and $0<\left(1-\frac{Q^*_w}{Q_w}\right)\le \frac{1}{R_{0w}}$ when $Q^*_w<Q_w\le K$. From equation \eqref{Eq40}, when $\left(1-\frac{Q^*_w}{Q_w}\right)<0$, we set $R_{0w}^*=\max\left(\frac{K(F_w+F_{\bar{w}})}{R_{0w}(K-Q)(F_w+\left(1-\rho_{_2}\right)F_{\bar{w}})}\right)$ and with $\left(2-\frac{Q_wF^*_w}{Q^*_wF_w}-\frac{Q^*_wF_w}{Q_wF^*_w}\right) \le 0$ implies $\frac{dV}{dt}<0$ for all $Q_w < {Q}^*_w$. When $0<\left(1-\frac{Q^*_w}{Q_w}\right)\le \frac{1}{R_{0w}}$, $\frac{dV}{dt}<0$ since 
$\max\left(\frac{(F_w+\left(1-\rho_{_2}\right)F_{\bar{w}})}{F_w+F_{\bar{w}}}\left(1-\frac{Q}{K}\right)\right)<1$. Hence, it follows from the Krasovkii-Lasalle theorem \cite{lasalle1960some, krasovskii1959nekotorye} that
\begin{equation} \label{Eq41} 
\left(Q_w,\ F_w\right)\to \left(Q^*_w,\ F^*_w\right)\ \ \ \ \ \mathrm{as}\ \ t\to \infty . 
\end{equation} 
It remains to show that $\left(Q_{\bar{w}},\ F_{\bar{w}}\right)\to \left(0,\ 0\right)\ \mathrm{as}\ t\to \infty $. In this case, ${\mathop{\mathrm{lim}\mathrm{}\mathrm{sup}}_{t\to \infty } Q_w=\ }Q^*_w$ and ${\mathop{\mathrm{lim}\mathrm{}\mathrm{sup}}_{t\to \infty } F_w=\ }F^*_w$. Hence, there exists a sufficiently small number $\epsilon >0$ and $t_1>0$, such that ${\mathop{\mathrm{lim}\mathrm{}\mathrm{sup}}_{t\to \infty } F_w\le \ }F^*_w+\epsilon $  and ${\mathop{\mathrm{lim}\mathrm{}\mathrm{sup}}_{t\to \infty } Q_w\le \ }Q^*_w+\epsilon $ for all $t>t_1$. It follows from equation \eqref{Eq7} that for$\ t>t_1$,
\begin{equation} \label{Eq42} 
\frac{{dQ}_{\bar{w}}\left(t\right)}{dt}\le \left[\frac{{\phi }_{\bar{w}}F^{\infty 2}_{\bar{w}}+\rho_{_2}{\phi }_w\left(F^*_w+\epsilon \mathrm{\ }\right)F^{\infty }_{\bar{w}}}{\left(F^{\infty }_{\bar{w}}+F^*_w+\epsilon \mathrm{\ }\right)}\right]\left[1-\frac{Q^*_w+\epsilon }{K}\right]-{(\mu }_a+\psi )Q_{\bar{w}}(t), 
\end{equation} 
where $F^{\infty }_{\bar{w}}={\mathop{\mathrm{lim}\mathrm{}\mathrm{sup}}_{t\to \infty } F_{\bar{w}}\left(t\right)}$. Hence, by the comparison theorem \cite{smith1995theory} and letting$\ \epsilon \to 0$
\begin{equation} \label{Eq43} 
Q^{\infty }_{\bar{w}}={\mathop{\mathrm{lim}\mathrm{}\mathrm{sup}}_{t\to \infty } Q_{\bar{w}}\left(t\right)\ }\le \frac{{\phi }_{\bar{w}}F^{\infty 2}_{\bar{w}}+\rho_{_2}{\phi }_wF^*_wF^{\infty }_{\bar{w}}}{(F^{\infty }_{\bar{w}}+F^*_w){(\mu }_a+\psi )R_{0w}}. 
\end{equation} 
If $F^{\infty }_{\bar{w}}={\mathop{\mathrm{lim}\mathrm{}\mathrm{sup}}_{t\to \infty } F_{\bar{w}}\left(t\right)\ }=0$ then $Q^{\infty }_{\bar{w}}\le 0$. Otherwise,
\begin{equation} \label{Eq44} 
{\mathop{\mathrm{lim}\mathrm{}\mathrm{sup}}_{t\to \infty } F_{\bar{w}}\left(t\right)\ }\le \frac{\psi Q^{\infty }_{\bar{w}}}{2{\mu }_{\bar{w}}\ } 
\end{equation} 
and                                                                                                                   
\begin{equation} \label{Eq45} 
{Q^{\infty }_{\bar{w}}}^2\left(\frac{\psi }{2{\mu }_{\bar{w}}}\left[1-\frac{R_{0\bar{w}}}{R_{0w}}\right]\ \right)+Q^{\infty }_{\bar{w}}F^*_w\left(\ 1-\frac{\rho_{_2}{\mu }_w\ }{{\mu }_{\bar{w}}}\right)\le 0\  
\end{equation} 
Hence,
\begin{equation} \label{Eq46} 
\frac{-2{\mu }_{\bar{w}}F^*_w\mathrm{\ }\left(1-\frac{\rho_{_2}{\mu}_w\ }{{\mu }_{\bar{w}}}\right)}{\psi \left(1-\frac{R_{0\bar{w}}}{R_{0w}}\right)}\le Q^{\infty }_{\bar{w}} \le 0.
\end{equation} 
Thus, it is immediate that $Q^{\infty }_{\bar{w}}={\mathop{\mathrm{lim}\mathrm{sup}}_{t\to \infty} Q_{\bar{w}}\left(t\right)\ }=0$ and ${\mathop{\mathrm{lim}\mathrm{sup}}_{t\to \infty } F_{\bar{w}}\left(t\right)\ }=0$. Hence,$\left(Q_{\bar{w}}, F_{\bar{w}}\right)\to \left(0,\ 0\right)$  as  $t\to \infty$. This concludes the proof.
\end{proof}

It can be verified that ${\mu }_{\bar{w}}>\rho_{_2} {\mu }_w$ implies${\ Q}_w\le Q^*_w\ $provided${\ Q}_w\left(0\right)\le Q^*_w$. ${\mu }_{\bar{w}}>\rho_{_2} {\mu }_w$ was also used in equation \eqref{Eq45} above. In Figure \ref{Fig3}, we examined theorem \eqref{th1} by simulating equations \eqref{Eq7}-\eqref{Eq10} with parameters that satisfy these conditions.
\begin{figure}[h]
\begin{center}
\includegraphics[width =0.7\linewidth]{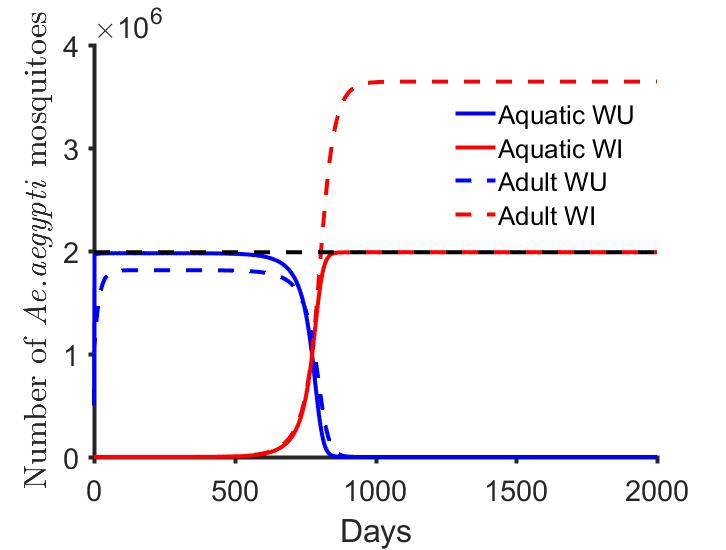}
\caption{\textbf{Global stability of \textit{Ae. aegypti} model \eqref{Eq7}-\eqref{Eq10}}. In this simulation, $R_{0w|\bar{w}}\boldsymbol{=}1.75$, $R_{0\bar{w}}\boldsymbol{=}91.67$, $R_{0w}\boldsymbol{=}169.23$, $R^*_{0w}=1.30$, $\rho_{_2} =0.05$, ${\mu }_{\bar{w}}=0.03$ (we assumed external factor to compensate for decrease in fitness), ${\mu }_w=0.07$, $K=2000000$ and the initial data are  $Q_{\bar{w}}\left(0\right)=500000,\ Q_w\left(0\right)=0,\ F_{\bar{w}}\left(0\right)=1000000,\ $ and $F_w\left(0\right)=1.$ The black dash line is the steady-state solution for $Q_w$.}
\label{Fig3}
\end{center}
\end{figure}
\subsection{Both mosquitoes}
 An interesting situation is to have both WI and uninfected mosquitoes in the \textit{Ae. aegypti} population. In such case, we will want the majority of the mosquitoes to be infected with \textit{Wolbachia}. For the systems of differential equations \eqref{Eq7}-\eqref{Eq10}, the co-existence equilibrium point is given as
\begin{equation} \label{Eq47} 
E_4=\left(d_1F^*_{\bar{w}}, d_2{F^*_w}, d_3{F^*_w},{F^*_w}\right) 
\end{equation} 
where, $F^*_w=\frac{K\psi }{2(\mu_{\bar{w}}d_3+\mu _{w})}\left[\frac{R_{0w}\left(1+\left(1-\rho_{_2} \right)d_3\right)-\left(1+d_3\right)}{R_{0w}\left(1+\left(1-\rho_{_2} \right)d_3\right)}\right]$,  $d_1=\frac{2{\mu }_{\bar{w}}}{\psi },d_2=\frac{{2\mu }_w}{\psi }$,
\newline $\ d_3=\frac{R_{0w|\bar{w}}({\mu }_{\bar{w}}-{\rho_{_2} \mu }_w)}{{\mu }_{\bar{w}}(1-\rho_{_2} )(1-R_{0\bar{w}|w})}$. From \eqref{Eq47}, it can be observed immediately that  $\rho_{_2} <1$, \newline $ R_{0w}>1$, $ R_{0\bar{w}}>1$ and either of these two conditions: 
\begin{enumerate}
\item  $ R_{0w|\bar{w}}<1 $ and ${\mu }_{\bar{w}}>\rho_{_2} {\mu }_w,$
\item  $R_{0w|\bar{w}}>1$ and ${\mu }_{\bar{w}}<\rho_{_2} {\mu }_w,$
\end{enumerate}
must be true for the existence of this equilibrium point. The conditions show key parameter relationships for both WI and WU mosquitoes to sustain themselves. If $ R_{0w|\bar{w}}<1 $, WU mosquitoes has a tolerable death rate that allows WI mosquitoes to survive and similar tolerance for $ R_{0w|\bar{w}}>1$. Hence all of the equilibrium points of the \textit{Wolbachia} spread model \eqref{Eq7}-\eqref{Eq10} can co-exist when $R_{0w|\bar{w}}<1\ $ \cite{xue2017two} and only the co-existence and WI-mosquito-only equilibrium points are locally stable when $\ R_{0w|\bar{w}}>1$. To establish whether this co-existence equilibrium point is stable or not, we use the general Jacobian expressions (see appendix \eqref{ap2}) to derive its corresponding characteristic equation:
\begin{equation} \label{Eq48} 
P\left(\lambda \right):={\lambda }^4+e_1{\lambda }^3+e_2{\lambda }^2+e_3\lambda +e_4=0, 
\end{equation} 
where the coefficients are given by the following expressions:
\begin{align} \label{Eq49} 
e_1 &=\left(F_1+{\mu }_{\bar{w}}\right)+\left(F_2+{\mu }_w\right), \\\label{Eq50} 
e_2 &= (F_2{\mu }_w-\frac{\psi B_2}{2}) +\left(F_1+{\mu }_{\bar{w}}\right)\left(F_2+{\mu }_w\right)+\left(F_1{\mu }_{\bar{w}}-\frac{\psi A_1}{2}\right)-T_1T_2, \\\label{Eq51} \nonumber
e_3 &= \left(F_2{\mu }_w-\frac{\psi B_2}{2}\right)\left(F_1+{\mu }_{\bar{w}}\right)+\left(F_2+{\mu }_w\right)\left(F_1{\mu }_{\bar{w}}-\frac{\psi A_1}{2}\right)\\
&-T_2\left(T_1{\mu }_w-\frac{\psi A_2}{2}\right)-{T}_1\left(T_2{\mu }_w-\frac{\psi B_1}{2}\right),\\  \label{Eq52} 
e_4 &=\left(F_1{\mu }_{\bar{w}}-\frac{\psi A_1}{2}\right)\left(F_2{\mu }_w-\frac{\psi B_2}{2}\right)-\left(T_1{\mu }_w-\frac{\psi A_2}{2}\right)\left(T_2{\mu }_{\bar{w}}-\frac{\psi B_1}{2}\right). 
\end{align} 

One way to establish the nature of the equilibrium point is to apply the Lienard and Chipart criterion \cite{lienard1914signe}, or the popular Routh-Hurwitz Criteria \cite{martcheva2015introduction}. For this criterion, it is necessary and sufficient to show that the coefficients of the quartic equation \eqref{Eq48} are greater than zero and that $e_1e_2e_3>e^2_3+e^2_1e_4\ $for the equilibrium point to be locally asymptotically stable. We investigate this equilibrium point by using a randomization method to show that there exists a parameter set for which the conditions above are satisfied and the quartic equation \eqref{Eq48} has negative roots (i.e. the eigenvalues of the associated Jacobian have negative real parts). For the parameters in Table 1, we sample 10000 parameter combinations assuming uniform distributions for the ranges listed in Table 1 and check whether the equilibrium point is asymptotically stable or not when condition (1) or (2) is satisfied. We found for the set of parameters satisfying condition (1) that the equilibrium point is unstable. For condition (2), the equilibrium point is locally asymptotically stable but with unrealistic parameter set (Figure \ref{Fig4}a). Despite the unrealistic nature of the parameter set satisfying condition (2) above, that equilibrium point is globally asymptotically stable as demonstrated numerically (Figure \ref{Fig4}b). Also, Table \ref{Table2} below lists the conditions for local asymptotic stability of the equilibrium points.

\begin{figure}[h]
\begin{center}
\includegraphics[width=1.0\linewidth]{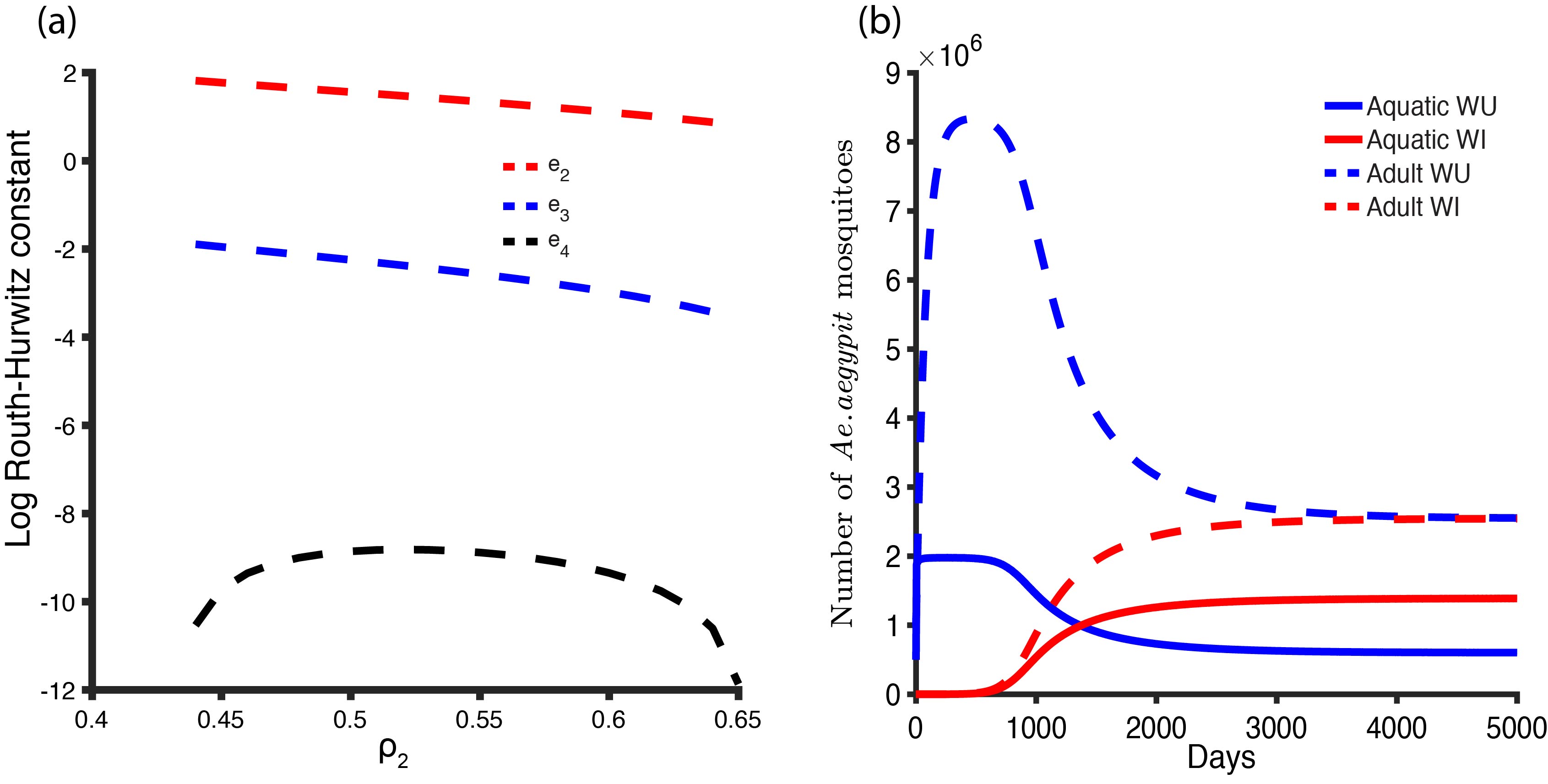}
\caption{\textbf{The stability conditions of the co-existence equilibrium point.} (a) The nature of the stability point changes with $\rho_{_2} .$ The Routh- Hurtwitz conditions are satisfied for condition (2) of the co-existence equilibrium points. (b) We set $R_{0w|\bar{w}}\boldsymbol{=}1.44$, $R_{0w}\boldsymbol{=}282.1$, $R_{0\bar{w}}\boldsymbol{=}97.6$, $\rho_{_2} =0.5$, ${\mu }_{\bar{w}}=0.03$, ${\mu }_w=0.013$, $K=2000000$ to show its global stability. The initial data are $Q_{\bar{w}}\left(0\right)=500000,\ Q_w\left(0\right)=0,\ F_{\bar{w}}\left(0\right)=1000000,\ $ and $F_w\left(0\right)=1.$}
\label{Fig4}
\end{center}
\end{figure}

\begin{table}[h]
\begin{center}
\caption{Conditions for stability of the equilibrium points}
\begin{tabular}{ p{2.0in} p{3.0in} } \hline 
\textbf{Equilibrium point} & \textbf{Stability conditions} \\ \hline 
$E_1$ (No\ Mosquitoes)  & $R_{0w}<1$ and $R_{0\bar{w}}<1$ \\
$E_{2}$ (Only WU Mosquitoes) & $R_{0w|\bar{w}}<1$ and $R_{0\bar{w}}>1$ \\ 
$E_{3}$ (Only WI Mosquitoes) & $R_{0w}>1$ and $\mu _{\bar{w}}<\rho_{2}\mu_{w}$ \\ 
$E_{4}$ (Both Mosquitoes) &  $\rho_{_2} <1$, $R_{0w|\bar{w}}>1,$  $R_{0w}>1,$ $R_{0\bar{w}}>1$ \newline and ${\mu }_{\bar{w}}<\rho_{_2} {\mu }_w.$ \\ \hline 
\end{tabular}
\label{Table2}
\end{center}
\end{table}
\subsection{\textbf{\textit{Wolbachia} invasive model with $\mathbf{\rho_{_1}\in (0,1]}$ and $\sigma>0$}}
When $\rho_{_1} \in (0, 1]$ and $\sigma>0$, the \textit{Wolbachia} invasive model \eqref{Eq7} - \eqref{Eq10} has three steady state solutions: $P_1=(0,0,0,0)$, $P_2=\left(Q^*_{\bar{w}},0,\ F^*_{\bar{w}},\ 0\right),\ $where $Q^*_{\bar{w}}\boldsymbol{=}K\left(1-\frac{1}{R_{0\bar{w}}}\right),\ F^*_{\bar{w}}\boldsymbol{=}\frac{\psi K}{2{\mu }_{\bar{w}}}\left(1-\frac{1}{R_{0\bar{w}}}\right),$ and $P_3=\left(Q^*_{\bar{w}},Q^*_w,\ F^*_{\bar{w}},F^*_w\right)$ with the expression for terms in $P_3$ defined later. Here, we do not have the WI-only mosquito equilibrium point because of the per capita loss of \textit{Wolbachia} infection rate ($\sigma$) that always replenishes the WU population.

The adjusted invasive reproductive number is:
\begin{equation} \label{Eq50a} 
R^1_{0w|\bar{w}}=\frac{R_{0w}\left(1-{\rho }_2\right){\mu }_w}{R_{0\bar{w}}({\mu }_w+\sigma )}. 
\end{equation} 
As expected, the expression for the invasive reproductive number shows that the loss of \textit{Wolbachia} infection reduces $R_{0w|\bar{w}}$ and in turn, the rate at which the WI mosquitoes invade the WU population. The imperfect maternal transmission between adult WI mosquitoes does not affect the adjusted reproductive number.

First, we investigate the impact of the individual reproduction numbers on the dynamics of the general \textit{Wolbachia} invasive model \eqref{Eq7} - \eqref{Eq10}. As before, when both $R_{0w}$ and $R_{0\bar{w}}$ are less than one the solutions tend to the no-mosquitoes equilibrium point (Figure \ref{Fig6}a). Similar to subsection \eqref{sub1}, the Jacobian for the steady-state solution $P_2$ is
\begin{equation} \label{Eq51b} 
J\left(E_2\right)=\left( \begin{array}{cccc}
-\left({\mu }_a+\psi \right)R_{0\bar{w}} & \left({\mu }_a+\psi \right)({1-R}_{0\bar{w}}) & \frac{{\phi }_{\bar{w}}}{R_{0\bar{w}}\ } & \frac{\left({{\rho }_2{\phi }_w-\phi }_{\bar{w}}\right)}{R_{0\bar{w}}} \\ 
0 & -\left({\mu }_a+\psi \right) & 0 & \frac{{\phi }_w\left(1-{\rho }_2\right)}{R_{0\bar{w}}} \\ 
\frac{\psi }{2} & 0 & -{\mu }_{\bar{w}\ \ \ \ } & \sigma  \\ 
0 & \frac{\psi }{2} & 0 & -({\mu }_w+\sigma ) \end{array}
\right), 
\end{equation} 
and the point is locally asymptotically stable whenever $R^1_{0w|\bar{w}}<1$ and $R_{0\bar{w}}>1$. We state the following theorem to show the conditions for global stability of $P_2$
\begin{theorem} \label{th2}
Whenever $\sigma +{\rho }_2{\mu }_w\ge {\mu }_{\bar{w}}\left(1-\ {\rho }_1\right)$,  the WU mosquito-only equilibrium point ($P_2$) is the only non-trivial equilibrium point whenever $R^1_{0w|\bar{w}}<1\ \mathrm{and}\ R_{0\bar{w}}>1$.
\end{theorem}
\begin{proof}
For the equilibrium point $P_3$, we have;
\begin{equation} \label{Eq52b} 
Q^*_w=\frac{2\left(\sigma +{\mu }_w\right)F^*_w}{\psi }, 
\end{equation} 
\begin{equation} \label{Eq53} 
Q^*_{\bar{w}}=\frac{2\left({\mu }_{\bar{w}}F^*_{\bar{w}}-\sigma F^*_w\right)}{\psi },
\end{equation} 
and using Equation \eqref{Eq7} and \eqref{Eq8}, we have
\begin{equation} \label{Eq54} 
\frac{{\phi }_{\bar{w}}F^{*2}_{\bar{w}}+{{{\rho }_1\phi }_wF}^{*2}_w+{\rho }_2{\phi }_wF^*_wF^*_{\bar{w}}}{{{{(1-\rho }_1)\phi }_wF}^{*2}_w+{(1-\rho }_2){\phi }_wF^*_wF^*_{\bar{w}}}=\frac{Q^*_{\bar{w}}}{Q^*_w}\  
\end{equation} 
From \eqref{Eq52} and \eqref{Eq53}, equation \eqref{Eq54} becomes
\begin{equation} \label{Eq55} 
B_1F^{*2}_w+B_2F^*_{\bar{w}}F^*_w+B_3F^{*2}_{\bar{w}}=0,                                                        
\end{equation} 
where 
\begin{align}\nonumber
B_1&={\phi }_w\left(\sigma +{\rho }_1{\mu }_w\right),\\\nonumber
B_2&=\left(\sigma +{\rho }_2{\mu }_w-{\mu }_{\bar{w}}\left(1-\ {\rho }_1\right)\right){\phi }_w, \\\nonumber
B_3&=\left(1-R^1_{0w|\bar{w}}\right)({\mu }_w+\sigma ){\phi }_{\bar{w}}.
\end{align}
Thus if $\sigma +{\rho }_2{\mu }_w\ge {\mu }_{\bar{w}}\left(1-\ {\rho }_1\right)$ and $R^1_{0w|\bar{w}}>1$, the quadratic equation \eqref{Eq55} has a positive solution: $F^*_w=C_1F^*_{\bar{w}}$ , with 
\begin{equation}\label{Eq56}
C_1=\frac{-B_2+\sqrt{B^2_2-4B_1B_3}}{2B_1}.
\end{equation}
and 
\begin{equation}\label{Eq55a}
F_{\bar{w}}^* = \frac{\psi K}{2(C_1 \mu_w +\mu_{\bar{w}})}\left( 1-\frac{2(1+C_{1})(\mu_{a}+\psi)(\mu_{\bar{w}}-\sigma C_{1})}{\psi(\phi_{\bar{w}}+\phi_{w}C_{1}(\rho_{1}+\rho_{2}))}\right).
\end{equation}
Hence, when $\sigma +{\rho }_2{\mu }_w\ge {\mu }_{\bar{w}}\left(1-\ {\rho }_1\right)$ and $R^1_{0w|\bar{w}} < 1$ we get a negative solution for$\ F^*_w$. Since our general model is biological meaningful (see \ref{ap1}), $P_3$ cannot exist for these conditions. That leaves $P_2$ has the only non-trivial equilibrium point.
\end{proof}
\begin{figure}[h]
\begin{center}
\includegraphics[width=1.0\linewidth]{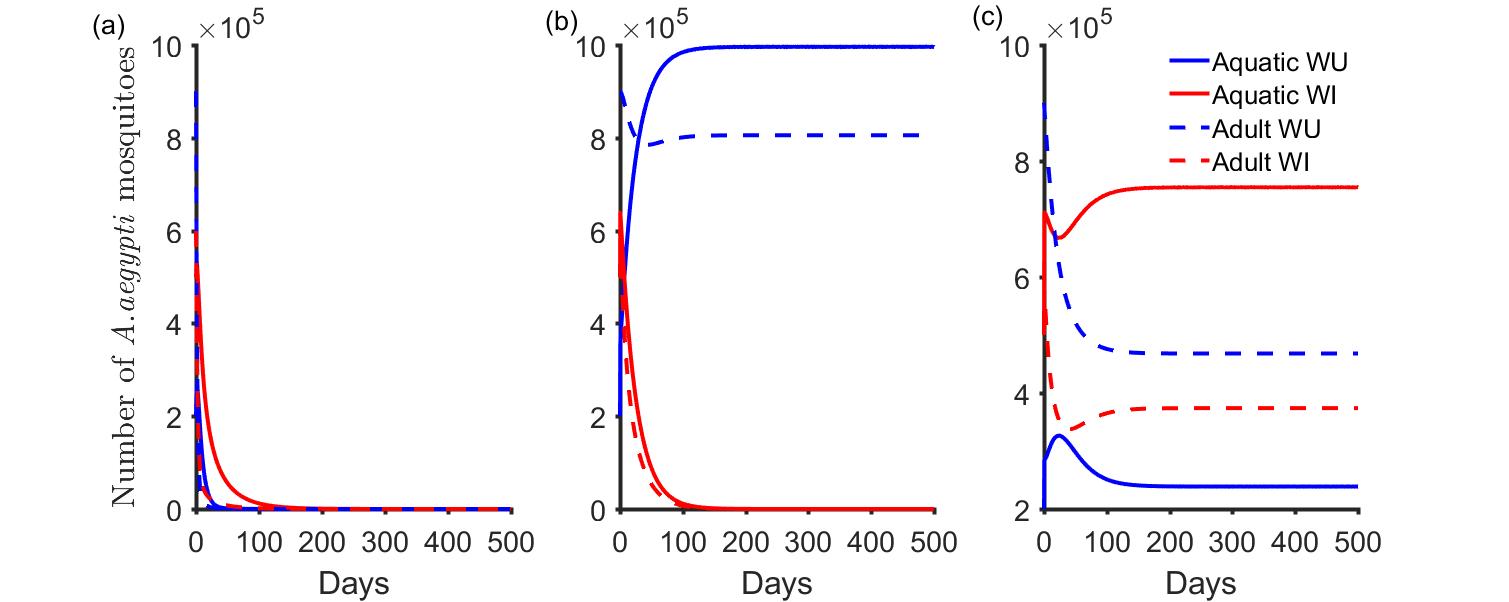}
\caption{\textbf{Simulation of the general \textit{Wolbachia} invasive model \eqref{Eq7} -\eqref{Eq10} for different steady state solutions.} \textbf{(a)} Here, we set $R_{0w}=0.86$, $R_{0\bar{w}}=0.58$, $\rho_{_1}=0.01$ and $\sigma =0.04$. \textbf{(b)} The general reproduction number ($R^1_{0w|\bar{w}}$) is $0.51$, $R_{0w}= 281.8$, $R_{0\bar{w}}=252.76$, $\rho_{_1}=0.05$ and $\sigma =0.07$.\textbf{(c)} $R^1_{0w|\bar{w}}>1$, $R_{0w}= 281.76$, $R_{0\bar{w}}=101.13$, $\rho_{_1}=0.04$ and $\sigma =0.05$.}
\label{Fig6}
\end{center}
\end{figure}

Figure \ref{Fig6}b shows the numerical demonstration of theorem \ref{th2}. The condition  $\sigma +{\rho }_2{\mu }_w\ge {\mu }_{\bar{w}}\left(1-\ {\rho }_1\right)$  can be interpreted as a condition that determines the mosquito fitness advantage of WU mosquitoes over WI mosquitoes. If $\sigma +{\rho }_2{\mu }_w<{\mu }_{\bar{w}}\left(1-\ {\rho }_1\right)$, either of the following two conditions guarantees a positive solution for$\ F^*_w$;
\begin{enumerate}
\item  $1-\frac{{\mu }_w{\left(\sigma +{\rho }_2{\mu }_w-{\mu }_{\bar{w}}\left(1-\ {\rho }_1\right)\right)}^2R_{0w}}{4{\mu }_{\bar{w}}\left({\mu }_w+\sigma \right)(\sigma +{\rho }_1{\mu }_w)R_{0\bar{w}}}<R^1_{0w|\bar{w}}<1$ 
\item  $R^1_{0w|\bar{w}}>1$,
\end{enumerate}
and if $R^1_{0w|\bar{w}}>1$ and $R_{0\bar{w}}>1$  leave $P_3$ as the only possible stable point in the positive quadrant $\mathbb{R}_{+}^4$. If this point is locally asymptotically stable and no other solutions exist in the plane then it is globally asymptotically stable for any positive initial condition (see Poincar\'{e}-Bendixson Trichotomy theorem and Figure \ref{Fig6}c).  The condition $\sigma +\rho_{_2}\mu_w < \mu_{\bar{w}}\left(1- \rho_{_1}\right)$ is the condition for backward bifurcation (Figure \ref{Fig7}a) and it implies that $P_3$ has two equilibria points with the one with higher $F^*_w$ locally asymptotically stable (Figure \ref{Fig7}b and c) with $R_{0}^* = \frac{\mu_{w}\left(\sigma +\rho _{_2}\mu_{w} - \mu_{\bar{w}}\left(1-\rho_{_1}\right)\right)^2 R_{0w}}{4\mu_{\bar{w}}\left(\mu_{w} +\sigma \right)(\sigma +\rho_{_1}\mu _w)R_{0\bar{w}}}$. The derivation of the conditions for local stability of the point $P_3$ when $\sigma +\rho_{_2}\mu_w < \mu_{\bar{w}}\left(1- \rho _{_1}\right)$ is shown in appendix \ref{ap3}.

\begin{figure}[h]
\begin{center}
\includegraphics[width=1.0\linewidth]{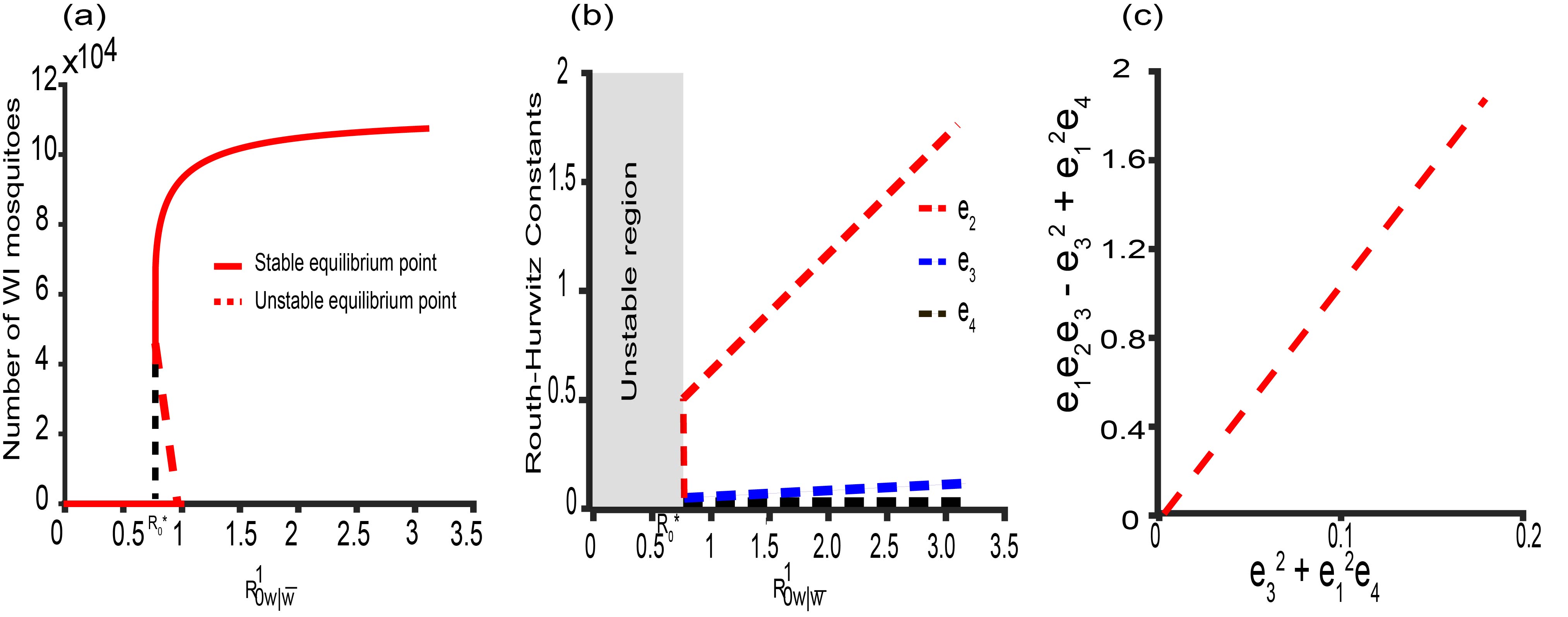}
\caption{\textbf{The backward bifurcation and local stability of the steady-state solution $P_3$}. We vary $\phi_{w}$ and set  values the following key parameters as: $R_{0\bar{w}}=122.5$, $\rho_{_1}=0.1$, $\rho_{_2}=0.06$ and $\sigma =0.02$. (a) Shows the backward bifurcation of the general model association with equilibrium points $P_2$ and $P_3$. (b) and (c) establish the local stability of the equlibrium point $E_4$ using Routh-Hurwitz conditions - (b), $e_1>0$, $e_2>0$, $e_3>0$, $e_4>0$ and (c),$e_1 e_2 e_3 > e_3^2+e_1^2 e^4$.}
\label{Fig7}
\end{center}
\end{figure}
\begin{figure}[h]
\begin{center}
\includegraphics[width=0.9\linewidth]{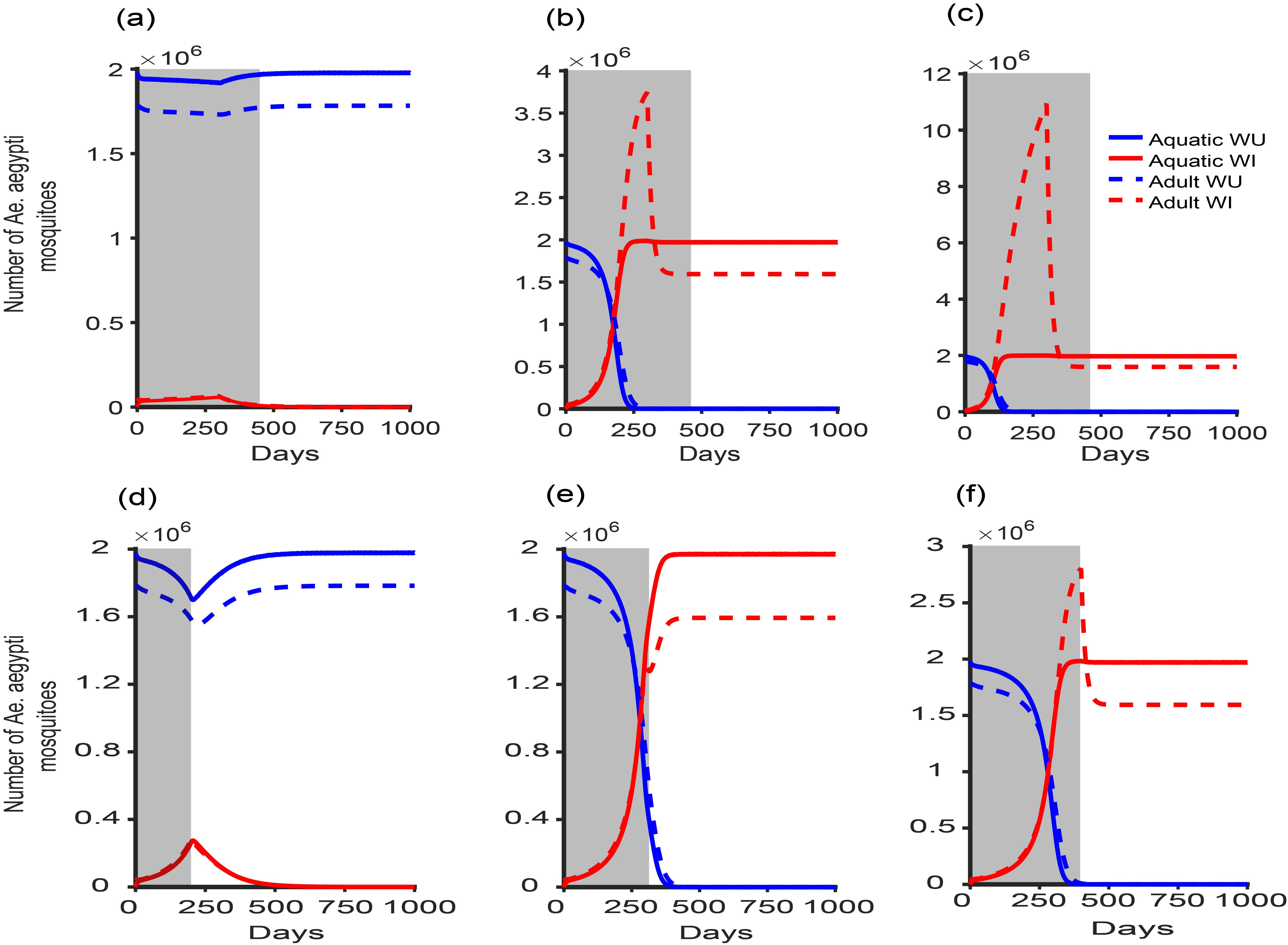}
\caption{\textbf{The \textit{Wolbachia} constant release rate program for different $\delta$ and time frames}. With final time of 480 days as shown by the grey area, (a) $\delta = 0.02$, (b) $\delta = 0.04$, and (c) $\delta = 0.06$. Setting $\delta = 0.03$, we varied the final time (grey area) for (d) $t = 200$ days, (e) $t = 300$ days, and (f) $t = 400$ days. The other parameter values are as in Table \ref{Table1}. Note that $Q_w(0) = 0$, $F_w(0) = \delta F_{\bar{w}}(0),$ and assuming that WU mosquitoes are in the WU-only equilibrium point at $t = 0$.}
\label{Fig8}
\end{center}
\end{figure}
\section{Optimal \textit{Wolbachia} release problem}
Similar to \cite{campo2018optimal}, we considered two release strategies: constant release rate and variable release rate. In \cite{campo2018optimal}, the focus is to wipeout WU mosquitoes and Rafikov et al. \cite{Rafikov2019} strategy is to have more WI mosquitoes than WU mosquitoes. Our modelling work has shown that the focus could be different depending on the dynamics of the \textit{Wolbachia} strain in the \textit{Ae. aegypti} population.
\subsection{Release strategy when $\rho_1 = 0$ and $\sigma =0$}\label{sub12}
For this case, the obvious strategy is to replace the WU mosquitoes with the infected ones. The parameter sets to ensure the conditions for the existence of the co-existence equilibrium in this case are unrealistic and it will be very difficult to achieve as there are  limitations on the parameters that can be controlled. Adjusting equation \eqref{Eq10} for the constant release rate strategy, we have
\begin{equation}\label{Eq62}
\frac{dF_w}{dt}=\frac{\psi }{2}Q_w -(\mu _w - \delta)F_w, 
\end{equation}
where $\delta$ is the per capita release rate. We need to force $R_{0w|\bar{w}}>1$ and $\mu_{\bar{w}}>\rho_2 (\mu_w -\delta)$. With this, the only locally asymptotically stable point is WI-only mosquito point. The two prior conditions indicate that :
\begin{equation}\label{Eq63}
\max\left(\mu_w - \frac{\mu_{\bar{w}}}{\rho_2},\mu_w - \frac{\phi_w\mu_{\bar{w}}(1-\rho_2)}{\phi_{\bar{w}}}\right)<\delta \le \delta_{max}
\end{equation}
Here, we set $\delta_{max} = 0.068$, the value of $\mu_w$ in Table \ref{Table1}. However, it can be large as desired depending on the resource constraint.
Using the values in Table \ref{Table1}, we have $ 0.019\le\delta \le 0.068$. This bound on $\delta$ is novel as we know the allowable range to acheive our goal. Figure \eqref{Fig8} shows the application of the constant release rate program for different rates for 480 days (same as the release program in Townsville, Australia \cite{o2018scaled}) and fixing $\delta = 0.03$ for different time frames. In Figure (\ref{Fig8}a), the goal is not achieved as WI mosquitoes are immediately wipeout and for Figure (\ref{Fig8}d, e, and f), it takes time for  WI mosquitoes to establish themselves.

For the variable release rate, we follow similar approach as \cite{campo2018optimal} by introducing control variable $u(t) \in [\max\left(\mu_w - \frac{\mu_{\bar{w}}}{\rho_2},\mu_w - \frac{\phi_w\mu_{\bar{w}}(1-\rho_2)}{\phi_{\bar{w}}}\right), \mu_w ]$. Hence, the variable release problem is an optimal control problem with constraint on the both end points: 
\begin{equation}\label{Eq64}
   \minimise \,\, J(u) =\int_{0}^{t_f}(c_1u(t)F_w(t)+c_2u^2(t))dt
\end{equation}
subject to equations \eqref{Eq7} - \eqref{Eq9}, and 
\begin{equation}\label{Eq65}
\frac{dF_w}{dt}=\frac{\psi }{2}Q_w -(\mu _w - u(t))F_w, 
\end{equation} 
\begin{equation}\label{Eq66}
   F_{\bar{w}}(t_f) = 0
\end{equation} 
\begin{equation}\label{Eq67}
   F_{w}(0) = u(0) F_{\bar{w}}(0)
\end{equation}
 \begin{figure}[h]
\begin{center}
\includegraphics[width=1.0\linewidth]{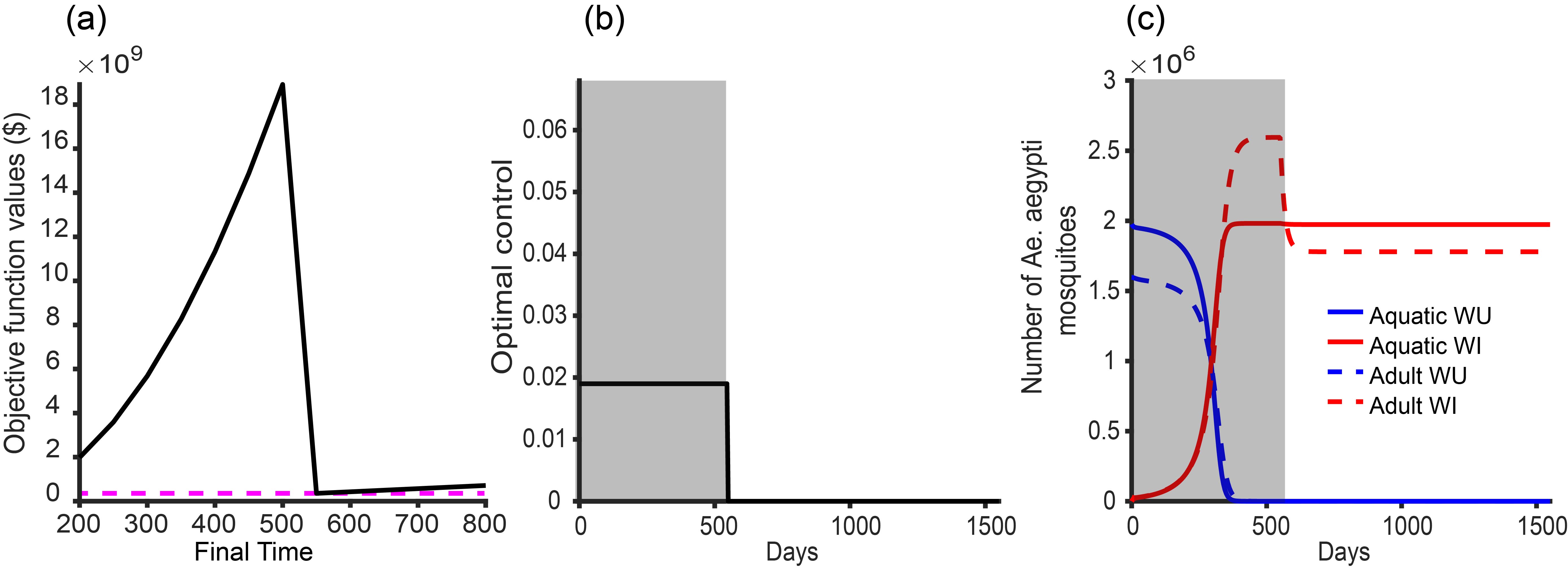}
\caption{\textbf{The \textit{Wolbachia} optimal variable rate release program}. (a) Objective function value for different final times. The minimal cost is at the final time, $t_f =550$, as indicated by the pink line, (b) The optimal control rate is set at the minimum value, and (c) The outcome of this program for  $t_f =550$. The novelty of the algorithm used in this computation is that $F_{\bar{w}}(t_f)$ need not to be zero at $t_f$ for successful replacement of WU mosquitoes. The parameter values used are in Table \ref{Table1} with WU mosquitoes at the WU-only equilibrium point at $t = 0$.The grey line in (c) is the release program period}
\label{Fig9}
\end{center}
\end{figure}
 $c_1u(t)F_w(t), (c_1 >0)$ is the cost per unit time associated wih this strategy and $B_2$ is a balancing cost. The quadratic term is to ensure we  have a regular optimal control. With the terminal constraint \eqref{Eq66}, the global stability of the WI-mosquito only equilibrium point (Theorem \ref{th1}) only requires $R_{0w}>1$. Following the release program of \textit{wMel} strain of \textit{Wolbachia} in Townsville, Australia \cite{o2018scaled}, stage 1 requires 14 months ($\approx$ 480 days) and it cost $\$69,732$ per km$^2$. Additionally, there is a cost associated with staff involved in the program. Here, we set $c_1 = \$69,732$ and $c_2 = 1$.
 and solve the optimal control problem \eqref{Eq64}-\eqref{Eq66} using the function space conjugate gradient algorithm \cite{edge1976function}. See \ref{ap4} for the optimal characterization and algorithm for this problem. Since $t_f$ is unknown, we solve the optimal control problem \eqref{Eq64}-\eqref{Eq67} for $t_f = [200, 800]$ with a step of 50 days and select the minimum objective function value. This approach is a modification of the  algorithm described in chapter 7, section 7.4  of \cite{bryson1975applied} for solving an optimal control problem with an unspecified final time. Figure \ref{Fig9} shows the optimal solutions for the variable release strategy. For this strategy, we do not need to wipe out all the WU mosquitoes to achieve our aim. 

\subsection{Release strategy when $\rho_1 \in (0, 1] $ and $\sigma > 0$}
For this case, we can only have a mix population and we will want more WI mosquitoes. Hence, equation \eqref{Eq10} becomes
\begin{equation}\label{Eq68}
\frac{dF_w}{dt}=\frac{\psi }{2}Q_w -(\mu _w + \sigma - \delta )F_w, 
\end{equation}
for the constant release rate strategy. Since we want more WI mosquitoes, $ C_1 > 1$, $\mu_{\bar{w}}(1-\rho_{1})> \sigma+\rho_2(\mu_w -\delta)$ and $R^{1}_{0w|\bar{w}}>1$. Thus, $ \mu_w + \sigma- \frac{\phi_w\mu_{\bar{w}}(1-\rho_2)}{\phi_{\bar{w}}}<\delta < \delta_{max} $. Again, we set $\delta_{max}= 0.1$ and using values in Table \ref{Table1} with $\sigma = 0.02$ gives  $ 0.039\le \delta \le  0.1$. We varied $\delta$ and implement this release program for two years.  From Figure \ref{Fig10}, after stoping the release program, it does not take much time for the WU mosquitoes to regain their dominance. If this strategy is to be adopted it will require continous application until all arboviral infections are eliminated. 
\begin{figure}[h]
\begin{center}
\includegraphics[width=1.0\linewidth]{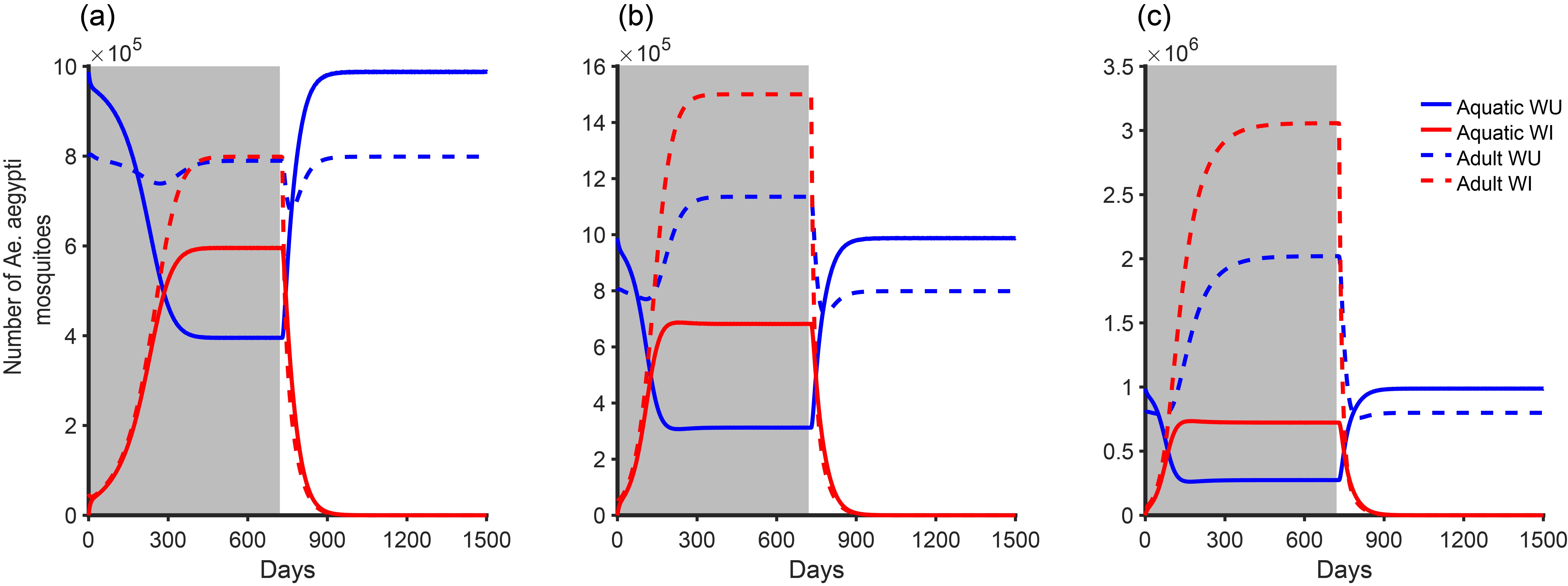}
\caption{\textbf{The \textit{Wolbachia} constant release program for different $\delta$ when $\rho_1 \in (0, 1] $ and $\sigma > 0$}. (a) $\delta = 0.06$, (b) $\delta = 0.076$ and  (c) $\delta = 0.1$. The other parameter values are as in Table \ref{Table1} and the grey area showing the period of release.}
\label{Fig10}
\end{center}
\end{figure}

For the optimal variable release rate strategy, we want at a particular time that the WI mosquitoes will be greater than the WU mosquitoes and maintains it. From Figure \ref{Fig10}, it shows that this will be an ongoing process and that we can make WI mosquitoes more abundant as we want (Figure \ref{Fig10}c) depending on the allowable budget. Hence, since the release program will be ongoing, the benefit of this program is well captured when we incorporate the human infection dynamics to see the level of reduction in arboviral infections as a result of this scheme. This is out of scope of this manuscript.

\section{Discussion and conclusion}
In this study, we developed and analysed a complex dynamical system of a two-type-mosquito population in the presence of imperfect maternal transmission and loss of \textit{Wolbachia} infection in order to determine the necessary and sufficient conditions for the propagation of \textit{Wolbachia} infection in an \textit{Ae. aegypti} population. We derived the invasive reproductive numbers with or without the adult WI mosquitoes losing  their \textit{Wolbachia} infection and established the conditions for local and global stability of the equilibrium points. We further adopted our models to determine the optimal release program that will ensure that WI mosquitoes replace or are become more abundant than the WU ones. Our analyses showed that mosquitoes with \textit{Wolbachia} infection can dominate, co-exist or die out depending on whether they are fitter than mosquitoes without \textit{Wolbachia} infection. The results showed clearly which factors and conditions are necessary and sufficient for WI mosquitoes to persist.

From our modelling, and consistent with other modelling works \cite{ndii2012modelling, xue2017two, ferguson2015modeling}, if WI mosquitoes are introduced in small numbers, the WI mosquitoes will not disrupt and outbreed WU ones. However, continuous introduction of WI mosquitoes for a particular period of time \cite{o2018scaled} will rescale the invasive reproduction number and increase it to above one where the WI-only mosquito equilibrium point is globally stable. A model in \cite{xue2017two} derived the conditions for WI mosquitoes to invade the \textit{Ae. aegypti} population and showed that the \textit{Wolbachia} can still spread despite the fact that the number of \textit{Wolbachia}-infected offspring due to WI adult mosquitoes in the next generation was less than one. By assuming an equal sex ratio between male and female \textit{Ae. aegypti} mosquitoes, imperfect maternal transmission and loss of \textit{Wolbachia} infection, we were able to find this threshold explicitly. This threshold does not guarantee replacement of uninfected mosquitoes as there are two possible steady-state solutions with one stable and the other unstable. However, \textit{Wolbachia}-infected mosquitoes can dominate if the defined invasive reproductive number is greater than one (this is not possible realistically except external measure such as deliberate introduction of WI mosquitoes is implemented) and other conditions stated in our results are satisfied. Then it is possible to completely replace the WU mosquitoes or have them in less proportion depending on the \textit{Wolbachia} infection dynamics in the \textit{Ae. aegypti} population as shown by the optimal control problem. 

Our study has some limitations that may affect our conclusions. One, we assumed that the ratio of male to female \textit{Ae. aegypti} mosquitoes is the same. This has been shown under a laboratory study and may not be necessarily true in a real-life situation \cite{arrivillaga2004}. Whatever the ratio of males to females in the number of eggs laid by either WU or infected female \textit{Ae. aegypti} mosquitoes, the main factors that determines \textit{Wolbachia} take-over is the proportion of WI eggs in the next generation and death rate. Two, most of the parameters in the associated invasive reproductive number are seasonally dependent \cite{ndii2016effect, yang2009assessing}. The dependency of key parameters on temperature is likely to affect WU and infected mosquitoes in a similar way. This means that we are likely to have \textit{Ae. aegypti} populations with any of three possibilities: without \textit{Wolbachia}; with \textit{Wolbachia}; and co-existence. Adverse conditions for the female \textit{Ae. aegypti} mosquitoes are likely to reduce the ability to reproduce and fertilize their eggs, and this is likely to push the population towards the no-mosquito equilibrium point rather than changing the proportion of the \textit{Ae. aegypti} population with \textit{Wolbachia} infection. Lastly, we have mimiced the transmission dynamics of the \textit{wMel} strain of \textit{Wolbachia} in this work. Other strains such as \textit{wMelPop} and \textit{wAlbB} have similar dynamics to the \textit{wMel} strain but with some variabilty in their reproductive advantage that affects the rate of introduction to ensure replacement \cite{xue2018comparing}. However, the \textit{wAu} strain does not have the advantage of CI but comes with a high virus transmision blocking potential \cite{ant2018wolbachia}. For the \textit{wAu} strain, our modelling is not applicable as the lack of CI implies that this strain has to be combined with another \textit{Wolbachia} strain that has the advantage of CI and superinfection \cite{ant2018wolbachia}. Further, the experimental modelling work by Ferguson et al \cite{ferguson2015modeling} showed different transmission settings that different strains of \textit{Wolbachia} can be adopted to reduce dengue infections. However, with the evident of loss of cytoplasmic incompatibility under field conditions \cite{ross2019loss} this needs to be revisited and our modelling work can be a template.

In general, our modelling work in this study complements existing works \cite{ndii2012modelling,xue2017two, Rafikov2019, campo2018optimal, Qu2018} and bridges the gap between alternative ways that WU mosquitoes may have advantages over the WI infected ones. Hence, controlling dengue epidemics and other arboviral infections with \textit{Wolbachia} is promising but implementing the strategy comes at a cost that requires careful evaluation. We have shown the potential outcomes of implementing such a strategy and the key parameters that could be targeted to achieve the desired objectives. Another question is what level of other vector control method is needed to aid \textit{Wolbachia} propagation so progress made will not be lost. If \textit{Wolbachia} has a strong and sustained effect in the \textit{Ae. aegypti} population, it remains to be seen whether the \textit{Ae. albopictus} will take over as a key vector agent for viral transmission or whether climate change can negate all gains from such an introduction. These are questions for future research.

\appendix
\section{Positivity and boundness of solutions}\label{ap1}
\begin{theorem}
For any given non-negative initial conditions, the solutions of \textit{Wolbachia} invasive model with $\rho_1 = 0$ and $\sigma = 0$ are non-negative for all $t\ge 0$ and bounded.
\end{theorem}
\begin{proof}
We prove by contradiction that whenever a solution enters the feasible region  ${\mathbb{R}}^4_+$ , it stays there forever. Consider the following four cases:

\begin{enumerate}
\item  there exists a first time $t_1>0$ such that whenever $Q_{\bar{w}}\left(t_1\right)=0,\ \frac{\ dQ_{\bar{w}}(t_1)}{dt\ }<0,\ Q_w\left(t\right)\ge 0,\ F_{\bar{w}}\left(t\right)\ge 0,\ F_w\left(t\right)\ge 0,\ 0\le t\le t_1$

\item  there exists a first time $t_2>0$ such that whenever $Q_w\left(t_2\right)=0,\ \frac{\ dQ_{\bar{w}}(t_2)}{dt\ }<0,\ Q_{\bar{w}}\left(t\right)\ge 0,\ F_{\bar{w}}\left(t\right)\ge 0,\ F_w\left(t\right)\ge 0,\ 0\le t\le t_2$

\item  there exists a first time $t_3>0$ such that whenever $F_{\bar{w}}\left(t_3\right)=0,\ \frac{\ dF_{\bar{w}}(t_3)}{dt\ }<0,\ Q_{\bar{w}}\left(t\right)\ge 0,\ Q_w\left(t\right)\ge 0,\ F_w\left(t\right)\ge 0,\ 0\le t\le t_3$

\item  there exists a first time $t_4>0$ such that whenever $F_w\left(t_4\right)=0,\ \frac{\ dF_w(t_4)}{dt\ }<0,\ Q_{\bar{w}}\left(t\right)\ge 0,\ Q_w\left(t\right)\ge 0,\ F_{\bar{w}}\left(t\right)\ge 0,\ 0\le t\le t_4$
\end{enumerate}
First, it can be shown that $Q(t)\le K$ provided $Q\left(0\right)<K$. Then, for the first case;
\[\frac{\ dQ_{\bar{w}}(t_1)}{dt\ }=\left[\frac{{\phi }_{\bar{w}}F^2_{\bar{w}}(t_1)+{\rho }_2{\phi }_wF_w\left(t_1\right)F_{\bar{w}}(t_1)}{F_{\bar{w}}(t_1)+F_w(t_1)}\right]\left(1-\frac{Q_w(t_1)}{K}\right)\ge 0,\] 
which contradicts that $\frac{\ dQ_{\bar{w}}(t_1)}{dt\ }<0$. For all the remaining cases, we have; 
\[\frac{dQ_w(t_1)}{dt}=\left[\frac{{{\phi }_WF}^2_w(t_1)+(1-{\rho }_2){\phi }_wF_wF_M(t_1)}{F_{\bar{w}}(t_1)+F_w(t_1)}\right]\left(1-\frac{Q_{\bar{w}}(t_1)}{K}\right)\ge 0,\] 
\[\frac{dF_{\bar{w}}(t_1)}{dt}=\frac{\psi }{2}Q_{\bar{w}}\left(t_1\right)\ge 0,\] 
\[\frac{dF_w(t_1)}{dt}=\frac{\psi }{2}Q_w\left(t_1\right)\ge 0.\] 
Hence, the solutions are non-negative for all future times given non-negative initial data. It remains to show that the solutions are bounded.
\end{proof}

\begin{corollary}
Let $\left(t\right)=Q_{\bar{w}}\left(t\right)+Q_w\left(t\right)+F_{\bar{w}}\left(t\right)+F_w(t)$ , there exists a constant $\tau >0$ such that ${\mathop{\mathrm{lim}\mathrm{}\mathrm{sup}}_{t\to \infty } M\left(t\right)\le \tau }$ . \end{corollary} 
\begin{proof}
Adding equations \eqref{Eq7} to \eqref{Eq10}, we have
\begin{align}\nonumber
\frac{dM}{dt}=&\left[\frac{{\phi }_{\bar{w}}F^2_{\bar{w}}+{\phi }_wF_wF_{\bar{w}}+{\phi }_wF^2_w}{F_{\bar{w}}+F_w}\right]\left(1-\frac{Q}{K}\right)-{\mu }_a{(Q}_{\bar{w}}+Q_w)-\frac{\psi }{2}{(Q}_{\bar{w}}+Q_w)\\\label{ape1}
&-{\mu }_{\bar{w}}F_{\bar{w}}-{\mu }_wF_{w}.
\end{align}                                                           
Since $Q_{\bar{w}}< \ K,\ Q_w< K$, then from equations \eqref{Eq9} and \eqref{Eq10}, $F_{\bar{w}}\le \frac{\psi K}{2{\mu }_1}$ and  $F_w\le \frac{\psi K}{2{\mu }_1}$  , where ${\mu }_1=\mathrm{min}\mathrm{}({{\mu }_{\bar{w}},\mu }_w,{\mu }_a)$. Thus, equation \eqref{ape1} becomes
\[\frac{dM}{dt}\le \frac{\psi K\left({\phi }_U+2{\phi }_w\right)}{{4\mu }_1}-{\mu }_1M.\] 
Hence, it follows from the inequality that there exists a constant $\tau $ such that
\[{\mathop{\mathrm{lim}\mathrm{}\mathrm{sup}}_{t\to \infty } M\left(t\right)\le \tau \ }.\] 
\end{proof}

\section{Jacobian expression of the \textit{Wolbachia} invasive model with $\rho_{_2}=0$ and $\sigma =0$}\label{ap2}
The general Jacobian of the models (7-10) is given as
\begin{equation}\label{ape2} 
J=\left( \begin{array}{cccc}
-F_1  & -T_1 & A_1 & A_2 \\ 
-T_{2} & -F_2 & B_1  & B_2 \\
\frac{\psi }{2} & 0 & -{\mu }_{\bar{w}} & 0\\ 
0 & \frac{\psi }{2} &  0 &-{\mu }_{w}
\end{array}
\right), 
\end{equation} 
where
\begin{equation} \label{ape3} 
T_1=\left(\frac{{{{\phi }_{\bar{w}}F}^2_{\bar{w}}}^*+{\rho \phi }_w{F_{\bar{w}}}^*{F_w}^*}{({F_{\bar{w}}}^*+{F_w}^*)K}\right), 
\end{equation} 
\begin{equation} \label{ape4} 
T_2=\left(\frac{{{{\phi }_wF}^2_w}^*+{(1-\rho )\phi }_w{F_{\bar{w}}}^*{F_w}^*}{({F_{\bar{w}}}^*+{F_w}^*)K}\right), 
\end{equation} 
\begin{equation} \label{ape5} 
A_1=\left(1-\frac{Q^*}{K}\right)\left[\left(\frac{{{{\phi }_{\bar{w}}F}^2_{\bar{w}}}^*+2{\phi }_{\bar{w}}{F_{\bar{w}}}^*{F_w}^*+{\rho \phi }_w{F^2_w}^*}{{{{(F}_{\bar{w}}}^*+{F_w}^*)}^2\ \ }\right)\right]\ge 0, 
\end{equation} 
\begin{equation} \label{ape6} 
A_2=-\left(1-\frac{Q^*}{K}\right)\left[\left(\frac{{F^2_{\bar{w}}}^*\left({\phi }_{\bar{w}}-{\rho \phi }_W\right)}{{{{(F}_{\bar{w}}}^*+{F_w}^*)}^2}\right)\right]\le 0\ ,\mathrm{\ if}\ \ {\phi }_{\bar{w}}>{\rho \phi }_w, 
\end{equation} 
\begin{equation} \label{ape7} 
B_1=-\left(1-\frac{Q^*}{K}\right)\left[\left(\frac{{\rho \phi }_w{F^2_w}^*}{{{{(F}_{\bar{w}}}^*+{F_w}^*)}^2}\right)\right]\le 0, 
\end{equation} 
\begin{equation} \label{ape8} 
B_2=\left(1-\frac{Q^*}{K}\right)\left[\left(\frac{{{{(1-\rho )\phi }_wF}^2_{\bar{w}}}^*+2{\phi }_w{F_{\bar{w}}}^*{F_w}^*+{\phi }_w{F^2_w}^*}{{{{(F}_{\bar{w}}}^*+{F_w}^*)}^2\ \ }\right)\right]\ge 0,                                                               
\end{equation} 
\begin{equation} \label{ape9} 
F_1=\left({\mu }_a+\psi +T_1\right), 
\end{equation} 
\begin{equation} \label{ape10} 
F_2=\left({\mu }_a+\psi +T_2\right). 
\end{equation} 

\section{Local stability of the equilibrium point $P_{3}$ when $\sigma +{\rho }_2{\mu }_w<{\mu }_{\bar{w}}\left(1- {\rho }_1\right)$}\label{ap3}
The Jacobian of the general model \eqref{Eq7}-\eqref{Eq10} is given as:
\begin{equation} \label{ape10} 
J\left(E_4\right)=\left( \begin{array}{cccc}
h_{11}  & h_{12} & h_{13} & h_{14} \\ 
h_{21} & h_{22} & h_{23} & h_{24} \\
\frac{\psi }{2} &  0  &-{\mu }_{\bar{w}} & \sigma\\
0 & \frac{\psi }{2} & 0  & -({\mu }_w+\sigma )
 \end{array}
\right),
\end{equation} 
where
\begin{equation} \label{ape11} 
h_{11}=-{(\mu }_a+\psi )\left[\frac{K-Q^*_w}{K-Q^*}\right] 
\end{equation} 
\begin{equation} \label{ape12} 
h_{12}=-\frac{{(\mu }_a+\psi )Q^*_{\bar{w}}}{K-Q^*} 
\end{equation} 
\begin{equation} \label{ape13} 
h_{13}=\left(1-\frac{Q^*}{K}\right)\left[\left(\frac{{\phi }_{\bar{w}}(1+2C_1)+{\phi }_w\left({\rho }_2-{\rho }_1\right)C^2_1}{{(1+C_1)}^2\ \ }\right)\right] 
\end{equation} 
\begin{equation} \label{ape14} 
h_{14}=\left(1-\frac{Q^*}{K}\right)\left[\left(\frac{{{\rho }_1\phi }_wC_1(C_1+2)+{\rho }_2{\phi }_w-{\phi }_{\bar{w}}}{{(1+C_1)}^2\ \ }\right)\right] 
\end{equation} 
\begin{equation} \label{ape15} 
h_{21}=-\frac{{(\mu }_a+\psi )Q^*_w}{K-Q^*} 
\end{equation} 
\begin{equation} \label{ape16}
h_{22}=-{(\mu }_a+\psi )\left[\frac{K-Q^*_{\bar{w}}}{K-Q^*}\right]
\end{equation}
\begin{equation} \label{ape17}
h_{23}=\left(1-\frac{Q^*}{K}\right)\left[\left(\frac{\left({\rho }_1-{\rho }_2\right){\phi }_wC^2_1}{{(1+C_1)}^2\ \ }\right)\right]
\end{equation}
\begin{equation} \label{ape18} 
h_{24}=\left(1-\frac{Q^*}{K}\right)\left[\left(\frac{{\phi }_w(\left(1-{\rho }_1\right)C_1(C_1+2)+\left(1-{\rho }_2\right))}{{(1+C_1)}^2\ \ }\right)\right].
\end{equation} 
When $F^*_w\ =0$ and $Q^*_w=0$, then the Jacobian \eqref{ape10} becomes the Jacobian expression\eqref{Eq51}. The characteristic equation is
\begin{equation} \label{ape19} 
P\left(\lambda \right):={\lambda }^4+e_1{\lambda }^3+e_2{\lambda }^2+e_3\lambda +e_4=0, 
\end{equation} 
where the coefficients are given by the following expressions:
\begin{align} \label{ape20} 
e_1 &=\sigma +{\mu }_w+{\mu }_{\bar{w}}-h_{11}-h_{22}>0,\\ \label{ape21} 
e_2 &=h_{11}h_{22}-\left(\sigma +{\mu }_w+{\mu }_{\bar{w}}\right)\left(h_{11}+h_{22}\right)-\frac{\psi \left(h_{13}+h_{24}\right)}{2}+{\mu }_{\bar{w}}(\sigma +{\mu }_w)-h_{12}h_{21}\\
\label{ape22} 
e_3 &=\left(\sigma +{\mu }_w+{\mu }_{\bar{w}}\right)\left(h_{11}h_{22}-h_{12}h_{21}\right)-{\mu }_{\bar{w}}\left(\sigma +{\mu }_w\right)\left(h_{11}+h_{22}\right)+\\\nonumber &\frac{\psi \left(h_{24}(h_{11}-{\mu }_{\bar{w}}\right)+h_{13}\left(h_{22}-\sigma -{\mu }_w\right)-h_{23}\left(h_{12}+\sigma \right)-h_{14}h_{21}}{2}\\\label{ape22}
e_4 &= {\mu }_{\bar{w}}\left(\sigma +{\mu }_w\right)\left(h_{11}h_{22}-h_{12}h_{21}\right)+\frac{{\psi }^2\left(h_{13}h_{24}-h_{14}h_{23}\right)}{4}\\ \nonumber
& \frac{\psi \sigma (h_{13}\left(h_{22}-h_{21}\right)+h_{23}\left(h_{11}-h_{12}\right))}{2}+\ \frac{\psi {\mu }_w(h_{13}h_{22}-h_{12}h_{23})}{2}\\\nonumber
&+\frac{\psi {\mu }_{\bar{w}}(h_{11}h_{24}-h_{14}h_{21})}{2}.
\end{align}
As before, we need to show that the coefficients of the characteristics equation are greater than zero and that $e_1e_2e_3>e^2_3+e^2_1e_4\ $for the equilibrium point to be locally asymptotically stable whenever $\sigma +{\rho }_2{\mu }_w<{\mu }_{\bar{w}}\left(1-\ {\rho }_1\right)$ and $R^1_{0w|\bar{w}}>R^*_0$.

\section{Optimal characterization and algorithm}\label{ap4}
 The optimal control problem \eqref{Eq64} - \eqref{Eq67} when $\rho_1 = 0$ and $\sigma=0$ is nonlinear problem with constraint on the initial time for $F_{w}$ and final time for $F_{\bar{w}}$. The algorithm for solving this problem especially when the final time is unspecified can be computational intensive. However, a commercialised package such as GPOPS-II can be used to solve this type of optimal control problem \cite{patterson2014gpops}. Here, instead of solving the problem directly, we solve the augmented problem defined by deriving the augmented cost functional\cite{edge1976function}:
 \begin{equation}\label{ap4_1}
   J(u) = 10(F_{\bar{w}}(t_f)^2 +(F_{w}(0)-u(0)F_{\bar{w}}(0))^2)+ \int_{0}^{t_f}(Cu(t)F_w(t)+u^2(t))dt
 \end{equation}
The value 10 is arbitrary as any value can be used. Hence, if the equality constraint is satified the augmented cost functional becomes the original cost functional. Hence, we have a Bolza problem without constraints on the state variables to solve rather than Lagrange problem with constraint on the state variables. Hence the Hamitonian function is defined as:
 \begin{equation}\label{ap4_2}
   H = Cu(t)F_w(t)+u^2(t) + \lambda_{Q_{\bar{w}}}f_1 + \lambda_{Q_{w}}f_2 +\lambda_{f_{\bar{w}}}F_1 + \lambda_{f_{w}}F_4.
 \end{equation}
Where,
\begin{align*}
f_1 =&\left[\frac{{\phi }_{\bar{w}}F^2_{\bar{w}}+{{\rho_{_1}\phi }_wF}^2_w+\rho_{_2}{\phi }_wF_wF_{\bar{w}}}{F_{\bar{w}}+F_w}\right]\left(1-\frac{Q}{K}\right)-{(\mu }_a+\psi )Q_{\bar{w}},\\
f_2 =&\left[\frac{{{(1-\rho_{_1})\phi }_wF}^2_w+(1-\rho_{_2}){\phi }_wF_wF_{\bar{w}}}{F_{\bar{w}}+F_w}\right]\left(1-\frac{Q}{K}\right)-{(\mu }_a+\psi )Q_w,\\
f_3=&\frac{\psi }{2}Q_{\bar{w}}+\sigma F_w-{\mu }_{\bar{w}}F_{\bar{w}},\\
f_4=&\frac{\psi }{2}Q_w-\sigma F_w-({\mu }_w-u(t))F_w
\end{align*}
By Pontryagin's minimum principle\cite{pontryagin2018mathematical}, the necessary conditions for optimality are
\begin{align}\label{ap4_3}
\frac{d\lambda_{Q_{\bar{w}}}}{dt} =&-\frac{\partial H}{\partial Q_{\bar{w}}},\\\label{ap4_4}
\frac{d\lambda_{Q_{w}}}{dt} =&-\frac{\partial H}{\partial Q_{w}},\\\label{ap4_5}
\frac{d\lambda_{F_{\bar{w}}}}{dt} =&-\frac{\partial H}{\partial F_{\bar{w}}},\\\label{ap4_6}
\frac{d\lambda_{F_{w}}}{dt} =&-\frac{\partial H}{\partial F_{w}},\\\label{ap4_7}
g(u) = &\frac{\partial H}{\partial u},
\end{align}
with the transversality conditions,
$\lambda_{Q_{\bar{w}}}(t_f)=0$, $\lambda_{Q_{w}}(t_f)=0$, $\lambda_{F_{\bar{w}}}(t_f)=20F_{\bar{w}}(t_f)$ and  $\lambda_{F_{w}}(t_f)=0$. With this formulation and
\begin{equation}\label{ap4_8}
u_{i+1}(t)= min\left(\delta_{max},max\left(u_{i+1}(t),\mu_w - \frac{\phi_w\mu_{\bar{w}}(1-\rho_2)}{\phi_{\bar{w}}}\right)\right),
\end{equation}
at each control evaluation step, the conjugate gradient algorithm \cite{edge1976function, lasdon1967conjugate} is adopted to solve the problem. Similar approach can be adopted for $\rho_1 \in (0, 1] $ and $\sigma > 0$.

\bibliographystyle{unsrtnat}
\bibliography{reference}

%\end{linenumbers}
\end{document}